\documentclass[fleqn,reqno]{amsart}
\usepackage{amssymb}
\usepackage{amsmath}
\usepackage{amsthm}

\usepackage{hyperref}
\hypersetup{colorlinks=true}

\newtheorem{theorem}{Theorem}[section]

\newtheorem{lemma}[theorem]{Lemma}
\newtheorem{corollary}[theorem]{Corollary}

\theoremstyle{definition}
\newtheorem{definition}[theorem]{Definition}
\newtheorem{remark}[theorem]{Remark}

\newcommand{\C}{\mathbb{C}}
\newcommand{\K}{\mathbb{K}}
\newcommand{\N}{\mathbb{N}}
\newcommand{\R}{\mathbb{R}}

\newcommand{\cB}{\mathcal{B}}
\newcommand{\cD}{\mathcal{D}}
\newcommand{\cG}{\mathcal{G}}
\newcommand{\cH}{\mathcal{H}}
\newcommand{\cHS}{\mathcal{HS}}
\newcommand{\cK}{\mathcal{K}}
\newcommand{\cL}{\mathcal{L}}
\newcommand{\cM}{\mathcal{M}}
\newcommand{\cP}{\mathcal{P}}

\newcommand{\cT}{\mathcal{T}}

\newcommand{\abs}[1]{\left\vert#1\right\vert}
\newcommand{\cond}{\;\vert\;\;\;}

\newcommand{\dx}{\mathrm{d}}
\newcommand{\eqdf}{\mathrel{\mathop:}=}

\newcommand{\grup}[2]{\mbox{$\left\{#1_{#2}\right\}_{#2
 \in\mathcal{M}(\Omega)}$}}
\newcommand{\Id}{I}
\newcommand{\inn}[2]{\left\langle#1,#2\right\rangle}
\newcommand{\Lc}{\mathrm{L}}
\newcommand{\limt}[2]{\underset{#1,\mu(#2)\to\infty}{d-\lim}}
\newcommand{\newt}[2]{\tbinom{#1}{#2}}
\newcommand{\norm}[1]{\left\Vert#1\right\Vert}
\newcommand{\per}{\mathrm{per}}
\newcommand{\rodz}[3]{\mbox{$\left\{#1_{#2,#3}\right\}_{(#2,#3)
 \in\mathcal{M}(\Omega)\times\mathbb{N}}$}}
\newcommand{\set}[1]{\left\{#1\right\}}
\newcommand{\sgn}{\mathrm{sgn}}
\newcommand{\Tr}{\mathrm{Tr}}

\begin{document}

\title[Contractions of product density operators]{Contractions of
product density operators of systems of identical fermions and
bosons}

\author{Wiktor Radzki}

\address{Faculty of Mathematics and Computer Science,
  Nicolaus Copernicus University,
  ul. Chopina $12 \slash 18$,
  87-100 Toru\'{n},
  Poland}

\email{wiktorradzki@yahoo.com}

\date{October 25, 2008}

\thanks{\emph{PACS numbers}: 02.30.Tb, 03.65.-w, 05.30.-d,
05.30.Fk, 05.30.Jp, 05.70.-a}

\keywords{Density operator (matrix); Partial trace;  Contraction of
operator; Reduced density operator; Fermion; Boson; Thermodynamic
limit}

\begin{abstract}
Recurrence and explicit formulae for contractions (partial traces)
of antisymmetric and symmetric products of identical trace class
operators are derived. Contractions of product density operators of
systems of identical fermions and bosons are proved to be
asymptotically equivalent to, respectively, antisymmetric and
symmetric products of density operators of a~single particle,
multiplied by a~normalization integer. The asymptotic equivalence
relation is defined in terms of the thermodynamic limit of
expectation values of observables in the states represented by
given density operators. For some weaker relation of asymptotic
equivalence, concerning the thermodynamic limit of expectation
values of product observables, normalized antisymmetric and
symmetric products of density operators of a~single particle
are shown to be equivalent to tensor products of density
operators of a~single particle.

This paper presents the results of a~part of the author's thesis
[W. Radzki, \emph{Kummer contractions of product density matrices
of systems of $n$ fermions and $n$ bosons} (Polish), MS thesis,
Institute of Physics, Nicolaus Copernicus University,
Toru\'{n}, 1999].
\end{abstract}

\maketitle

\section{Introduction}

This paper,
presenting the results of a~part of the author's
thesis~\cite{RadzkiUMK99}, deals with contractions (partial traces)
of antisymmetric and symmetric product density operators
representing mixed states of systems of identical
noninteracting fermions and bosons, respectively.

If $\cH$ is a~separable Hilbert space of a~single fermion (boson)
then the space of the \mbox{$n$-fer}mion (resp. \mbox{$n$-bo}son)
system is the antisymmetric (resp. symmetric) subspace
$\cH^{\wedge n}$
(resp. $\cH^{\vee n}$) of
$\cH^{\otimes n}.$
Density operators of \mbox{$n$-fer}mion (resp. \mbox{$n$-bo}son)
systems are identified with those defined on $\cH^{\otimes n}$ and
concentrated on  $\cH^{\wedge n}$ (resp. $\cH^{\vee n}$).

Recall that the expectation value of an observable represented by
a~bounded selfadjoint operator $B$ on given Hilbert space in
a~state described by a~density operator $\rho$ equals $\Tr\,B\rho.$
If $B$ is an unbounded selfadjoint operator on a~dense subspace of
given Hilbert space, instead of $B$ one can consider its spectral
measure $E_B(\Delta)$ (which is a~bounded operator) of a~Borel
subset $\Delta$ of the spectrum of $B.$ Then $\Tr\,E_B(\Delta)\rho$
is the probability that the result of the measurement of the
observable in question belongs to $\Delta$~\cite{vonNeumannS32}.

\mbox{$k$-par}ticle observables of \mbox{$n$-fer}mion and
\mbox{$n$-bo}son systems ($k<n$) are represented, respectively, by
operators of the form
 \begin{equation}{\label{KE}}
 \stackrel{\wedge}{\Gamma^{n}_{k}}B
 =A^{(n)}_{\cH}\left(B\otimes\Id^{\otimes(n-k)}\right)A^{(n)}_{\cH},
 \quad
 \stackrel{\vee}{\Gamma^{n}_{k}}B
 =S^{(n)}_{\cH}\left(B\otimes\Id^{\otimes(n-k)}\right)S^{(n)}_{\cH}
 \end{equation}
(multiplied by $\newt{n}{k}$), where $A^{(n)}_{\cH}$ and
$S^{(n)}_{\cH}$ are projectors of $\cH^{\otimes n}$ onto
$\cH^{\wedge n}$ and $\cH^{\vee n},$ respectively,
$I$ is the identity operator on $\cH$ and $B$ is a~selfadjoint
operator on $\cH^{\otimes k}$(see~\cite{KummerJMP67}).
Operators~\eqref{KE} are called \emph{antisymmetric} and
\emph{symmetric expansions of $B$}. In view of the earlier
remark it is assumed that $B$ is bounded. The expectation values of
observables represented by
 $\stackrel{\wedge}{\Gamma^{n}_{k}}B$ and
 $\stackrel{\vee}{\Gamma^{n}_{k}}B$
in states represented by \mbox{$n$-fer}mion and \mbox{$n$-bo}son
density operators $K$ and $G,$ respectively, can be expressed as
 \begin{equation}{\label{kontrKE}}
 \Tr\,K\stackrel{\wedge}{\Gamma^{n}_{k}}B =\Tr\,B\Lc^{k}_{n}K,
 \quad
 \Tr\,G\stackrel{\vee}{\Gamma^{n}_{k}}B =\Tr\,B\Lc^{k}_{n}G
 \end{equation}
(see~\cite[Eqs.~(1.7),~(3.19)]{KummerJMP67}), where
\mbox{$k$-par}ticle density operators $\Lc^{k}_{n}K$ and
$\Lc^{k}_{n}G$ are \emph{\mbox{$(n,k)$-con}\-tractions of $K$} and
\emph{$G$} (see Definition~\ref{kontr}), called also \emph{reduced
density operators}. Such operators were investigated by
Coleman~\cite{ColemanRMP63}, Garrod and Percus~\cite{GarrodJMP64},
and Kummer~\cite{KummerJMP67}.

In the present paper particular interest is taken in the case when
$K$ and $G$ are \emph{product density operators}, i.e. they are of
the form
\begin{equation}{\label{prod}}
 K =\frac{1}{\Tr\,\rho^{\wedge n}}\rho^{\wedge n},
 \quad
 G =\frac{1}{\Tr\,\rho^{\vee n}}\rho^{\vee n},
\end{equation}
where
 $\rho^{\wedge n} =A^{(n)}_{\cH}\rho^{\otimes n}A^{(n)}_{\cH},$
 $\rho^{\vee n} =S^{(n)}_{\cH}\rho^{\otimes n}S^{(n)}_{\cH},$
and $\rho$ is a~density operator of a~single fermion or boson,
respectively. The first objective of this paper is to find the
recurrence and explicit formulae for $\Lc^{k}_{n}K$ and
$\Lc^{k}_{n}G$ for $K$ and $G$ being, respectively, antisymmetric
and symmetric products of identical trace class operators,
including operators~\eqref{prod}. The explicit form of the
operators $\Lc^{k}_{n}K$ and $\Lc^{k}_{n}G$ proves to be quite
complex. However, they can be replaced by operators with simpler
structure if only the limiting values of
expectations~\eqref{kontrKE}, in the sense of the thermodynamic
limit, are of interest. The second objective of this paper is to
find that simpler forms of contractions $\Lc^{k}_{n}K$ and
$\Lc^{k}_{n}G$ for product density operators~\eqref{prod},
equivalent to the complete expressions in the thermodynamic limit.

The problems described above have been solved for $k=1,2$ by
Kossakowski and Ma\'{c}kowiak~\cite{KossakowskiRMP86}, and
Ma\'{c}kowiak~\cite{MackowiakPR99}. The formulae they derived were
exploited in calculations of the free energy density of large
interacting \mbox{$n$-fer}mion and \mbox{$n$-bo}son
systems~\cite{KossakowskiRMP86,MackowiakPR99}, as well as in the
perturbation expansion of the free energy density for the
\mbox{$M$-im}purity Kondo Hamiltonian~\cite{MackowiakPA97}. In the
case of investigation of many-particle interactions of higher
order~\cite{VolovikWS92,TarasewiczPC00,MackowiakPC00,SchneiderEL04},
or higher order perturbation expansion terms of the free energy
density, the expressions for
 $(\Tr\,\rho^{\wedge n})^{-1}\Lc^{k}_{n}\rho^{\wedge n}$
and
 $(\Tr\,\rho^{\vee n})^{-1}\Lc^{k}_{n}\rho^{\vee n}$
with $k\geq 3$ prove to be needed in the canonical and grand
canonical ensemble approach, which is the physical motivation
for the present paper.

The main results of this paper are Theorems~\ref{rek}, \ref{jawny},
\ref{glowne}, and~\ref{zmiana}.

\section{Preliminaries}{\label{prelim}}

In this section notation and terminology are set up.

\subsection{Basic notation}{\label{oznacz}}

Let $(\cH,\inn{\cdot}{\cdot})$ be a~separable Hilbert space over
$\C$ or $\R$. The following notation is used in the sequel.
\newline
$I$ -- the identity operator on $\cH,$
\newline
$\cB(\cH)$ -- the space of bounded linear operators on
$\cH$ with the operator norm $\norm{\cdot}$,
\newline
$\cT(\cH)$ -- the space of trace class operators on $\cH$ with
the trace norm $\Tr\abs{\cdot},$
\newline
$\cB^{\ast}(\cH)$ -- the space of bounded selfadjoint operators on
$\cH$,
\newline
$\cB^{\ast}_{\geq 0}(\cH)$ -- the set of nonnegative definite
bounded selfadjoint operators on $\cH$,
\newline
$\cD(\cH)$ -- the set of density operators (matrices) on $\cH,$ i.e.
 \begin{equation*}
 \cD(\cH) =\set{D\in\cT(\cH)\cond D =D^{\ast}, D\geq 0, \Tr\,D=1}.
 \end{equation*}

Set
 $\cH^{\otimes n} =\underbrace{\cH\otimes\cdots\otimes\cH}_{n}$
and denote by $S_{n}$ the group of permutations of the set
 $\set{1,\ldots,n}.$
Let
 $A^{(n)}_{\cH},S^{(n)}_{\cH} \in\cB(\cH^{\otimes n})$
be the projectors such that
 \begin{equation*}
 A^{(n)}_{\cH}(\psi_{1}\otimes\cdots\otimes\psi_{n})
 =\frac{1}{n!}\sum_{\pi\in S_{n}}\sgn\,\pi\,\psi_{\pi(1)}\otimes
 \cdots\otimes\psi_{\pi(n)},
 \end{equation*}
 \begin{equation*}
 S^{(n)}_{\cH}(\psi_{1}\otimes\cdots\otimes\psi_{n})
 =\frac{1}{n!}\sum_{\pi\in S_{n}}\psi_{\pi(1)}\otimes
 \cdots\otimes\psi_{\pi(n)}
 \end{equation*}
for every $\psi_{1},\ldots,\psi_{n}\in\cH.$ The closed linear
subspaces
 $\cH^{\wedge n} =A^{(n)}_{\cH}\cH^{\otimes n}$ and
 $\cH^{\vee n} =S^{(n)}_{\cH}\cH^{\otimes n}$ of
$\cH^{\otimes n}$ are called the \emph{antisymmetric} and
\emph{symmetric subspace}, respectively.

The \emph{antisymmetric} and \emph{symmetric product} of operators
 $B\in\cB\left(\cH^{\otimes k}\right),$
 $C\in\cB(\cH^{\otimes m})$
are defined as
 $B\wedge C =A^{(k+m)}_{\cH}(B\otimes C)A^{(k+m)}_{\cH}$
and
 $B\vee C =S^{(k+m)}_{\cH}(B\otimes C)S^{(k+m)}_{\cH},$
respectively.
It is assumed
 $B^{\wedge n} =\underbrace{B\wedge\ldots\wedge B}_{n},$
 $B^{\vee n} =\underbrace{B\vee\ldots\vee B}_{n},$
and
 $B^{\wedge 1} =B^{\vee 1} =B.$
Clearly, if $B\in\cB(\cH)$ then
 $B^{\wedge n} =A^{(n)}_{\cH}B^{\otimes n}
 =B^{\otimes n}A^{(n)}_{\cH},$
 $B^{\vee n} =S^{(n)}_{\cH}B^{\otimes n}
 =B^{\otimes n}S^{(n)}_{\cH},$
and if
 $B\in\cB^{\ast}_{\geq 0}(\cH)$ then
 $B^{\wedge n},B^{\vee n}\in\cB^{\ast}_{\geq 0}(\cH^{\otimes n}).$

Set $\R_{+} =[0,+\infty)$
and $\overline{\R}_{+} =\R_+\cup\set{+\infty}.$
The product of measures $\mu,$ $\mu_1$ is denoted by
$\mu\otimes\mu_{1}$ and
 $\mu^{\otimes n}$ stands for
 $\underbrace{\mu\otimes\cdots\otimes\mu}_{n}.$
In subsequent sections use is made of \emph{product integral
kernels}, described in Appendix~\ref{sectjad}.

\subsection{Contractions of operators}{\label{rozkontr}}

The definition and basic properties of contractions of operators
are now recalled for the reader's convenience. A~discussion of
these properties was carried out by
Kummer~\cite{KummerJMP67,KummerJMP70}.

Let $\cH$ be a~separable Hilbert space over the field $\K=\C$ or
$\R.$

 \begin{definition}{\label{kontr}}
 Let $k,n\in\N,$ $k<n,$ and $K\in\cT(\cH^{\otimes n}).$ Then the
 \emph{\mbox{$(n,k)$-con}traction of $K$} is the operator
  $\Lc^{k}_{n}K\in\cT(\cH^{\otimes k})$
 such that
  \begin{equation}{\label{kontr1}}
  \forall_{C\in\cB\left(\cH^{\otimes k}\right)} \colon
  \quad \Tr_{\cH^{\otimes n}}(C\otimes\Id^{\otimes(n-k)})K
  =\Tr_{\cH^{\otimes k}}C\Lc^{k}_{n}K.
  \end{equation}
 It is also assumed $\Lc^{n}_{n}K =K.$
 \end{definition}

 \begin{remark}{\label{popr}}
 The operator $\Lc^{k}_{n}K$ always exists and is defined uniquely
 by Eq.~\eqref{kontr1}. $\Lc^{k}_{n}K$ is a~partial trace of $K$
 with respect to the component $\cH^{\otimes (n-k)}$ of
  $\cH^{\otimes n}=\cH^{\otimes k}\otimes\cH^{\otimes (n-k)}.$
 If
  $\cH =\cH_{Y}\eqdf L^{2}(Y,\mu),$
 where the measure $\mu$ is separable and \mbox{$\sigma$-fi}nite,
 and $\cK$ is  a~product integral kernel of $K$
 (see Appendix~\ref{sectjad}) then
 $\Lc^{k}_{n}K$ has an  integral kernel $\cK_{0}$ given by
 formula~\eqref{parttr1},
 according to Lemma~\ref{parttr} and Corollary~\ref{red}.
 \end{remark}

Under the assumptions of Definition~\ref{kontr} one has
 $\Tr_{\cH^{\otimes k}}\Lc^{k}_{n}K =\Tr_{\cH^{\otimes n}}K,$
and if $p\in\N,$ $k<p<n,$ then
 $\Lc^{k}_{p}\left(\Lc^{p}_{n}K\right) =\Lc^{k}_{n}K.$
Moreover, if
 $K\in\cB^{\ast}(\cH^{\otimes n})$
then
 $\Lc^{k}_{n}K\in\cB^{\ast}(\cH^{\otimes k}),$
and if
 $K\in\cB^{\ast}_{\geq 0}(\cH^{\otimes n})$
then
 $\Lc^{k}_{n}K\in\cB^{\ast}_{\geq 0}(\cH^{\otimes k}).$

Contractions of density operators are called \emph{reduced density
operators}. Contractions preserve the Fermi and the Bose-Einstein
statistics of the contracted operator, i.e. for
 $K\in A^{(n)}_{\cH} \cT(\cH^{\otimes n})A^{(n)}_{\cH}$
and
 $G\in S^{(n)}_{\cH}\cT(\cH^{\otimes n}) S^{(n)}_{\cH}$
one has
 $\Lc^{k}_{n}K\in A^{(k)}_{\cH} \cT(\cH^{\otimes k})A^{(k)}_{\cH}$
and
 $\Lc^{k}_{n}G\in S^{(k)}_{\cH}\cT(\cH^{\otimes k})S^{(k)}_{\cH}.$
For such $K$ and $G$ Eqs.~\eqref{kontrKE} hold.

The following theorem is a~part of Coleman's
theorem~\cite{ColemanRMP63,KummerJMP67}.

 \begin{theorem}{\label{colemferm}}
 Let $n\in\N,$ $n\geq 2.$ For every (\mbox{$n$-fer}mion) density
 operator
  $D\in\cD\left(\cH^{\otimes n}\right),$
  $D =A^{(n)}_{\cH}DA^{(n)}_{\cH},$
 one has
  $\norm{\Lc^{1}_{n}D} \leq\frac{1}{n}\norm{D}.$
 \end{theorem}

\section{Recurrence and explicit formulae for contractions of
products of trace class operators}{\label{potkontr}}

In this section recurrence and explicit formulae
for contractions of antisymmetric and symmetric powers
of single particle operators are derived.

In the whole section use is made of the Hilbert space
$\cH_{Y} \eqdf L^{2}(Y,\mu)$ over the field $\K=\C$ or $\R,$ where
the measure $\mu$ is separable and \mbox{$\sigma$-fi}nite.

The following theorem in the case of $k=1,2$ was proved
in~\cite{KossakowskiRMP86,MackowiakPR99}.

 \begin{theorem}[Recurrence formulae]{\label{rek}}
 Let $\rho\in\cT(\cH_{Y}).$ If $k,n\in\N,$ $1<k<n,$ then
 \begin{align}{\label{rek1}}
 \nonumber \newt{n}{k}\Lc^{k}_{n}\rho^{\wedge n}
 & =\newt{n-1}{k-1}
 \left(\Lc^{k-1}_{n-1}\rho^{\wedge(n-1)}\right)\wedge \rho \\
 &\quad -\newt{n-1}{k}\left(\Lc^{k}_{n-1}\rho^{\wedge(n-1)}\right)
 \left(\Id^{\otimes(k-1)}\otimes \rho\right)A^{(k)}_{\cH_{Y}},
 \end{align}
 \begin{align}{\label{rek2}}
 \nonumber \newt{n}{k}\Lc^{k}_{n}\rho^{\vee n}
 & =\newt{n-1}{k-1}
 \left(\Lc^{k-1}_{n-1}\rho^{\vee(n-1)}\right)\vee \rho \\
 &\quad +\newt{n-1}{k}\left(\Lc^{k}_{n-1}\rho^{\vee(n-1)}\right)
 \left(\Id^{\otimes(k-1)}\otimes \rho\right)S^{(k)}_{\cH_{Y}},
 \end{align}
 and if $n\in\N,$ $n\geq 2,$ then
 \begin{equation}{\label{rek3}}
 n\Lc^{1}_{n}\rho^{\wedge n} =\left(\Tr\,\rho^{\wedge(n-1)}\right)\rho
 -(n-1)\left(\Lc^{1}_{n-1}\rho^{\wedge(n-1)}\right)\rho,
 \end{equation}
 \begin{equation}{\label{rek4}}
 n\Lc^{1}_{n}\rho^{\vee n} =\left(\Tr\,\rho^{\vee(n-1)}\right)\rho
 +(n-1)\left(\Lc^{1}_{n-1}\rho^{\vee(n-1)}\right)\rho.
 \end{equation}
 \end{theorem}
\begin{proof}
Let $\varrho\colon Y\times Y\to\K$ be a~product integral kernel of
$\rho.$  For every $m\in\N$ define the mapping
 $\varrho^{\wedge m}\colon Y^m\times Y^m\to\K$
by the formula
 \begin{equation*}
 \varrho^{\wedge m}
 \left(
 \begin{array}{c}
 x_1,\ldots,x_m \\
 y_1,\ldots,y_m
 \end{array}
 \right)
 =\det\left[
 \begin{array}{ccc}
 \varrho(x_1,y_1)&\cdots&\varrho(x_1,y_m) \\
 \vdots&\cdots&\vdots\\
 \varrho(x_m,y_1)&\cdots&\varrho(x_m,y_m)
 \end{array}
 \right].
 \end{equation*}
Then the mapping
 $\cK\colon Y^n\times Y^n\to\K$
given by
 \begin{equation*}
 \cK(x_{1},\ldots,x_{n},y_{1},\ldots,y_{n})
 =\frac{1}{n!}\varrho^{\wedge n}
 \left(
 \begin{array}{c}
 x_1,\ldots,x_n\\
 y_1,\ldots,y_n
 \end{array}
 \right)
 \end{equation*}
is a~product integral kernel of
 $\rho^{\wedge n} =A^{(n)}_{\cH_{Y}} \rho^{\otimes n}.$

Eq.~\eqref{rek1} will be first proved for $n>k+1.$
In view of Remark~\ref{popr}, an integral kernel
 $\cL\colon Y^k\times Y^k\to\K$
of
 $\newt{n}{k}\Lc^{k}_{n}\rho^{\wedge n}$
can be given by
 \begin{equation*}
 \cL(x^{\prime},y^{\prime})
 =\newt{n}{k}\int_{Y^{n-k}}
 \cK(x^{\prime},x^{\prime\prime},y^{\prime},x^{\prime\prime})
 \,\dx\mu^{\otimes(n-k)}(x^{\prime\prime})
 \end{equation*}
for \mbox{$\mu^{\otimes 2k}$-a.a.}
 $(x^{\prime},y^{\prime})\in Y^{k}\times Y^{k}.$
Performing $k!$ permutations of the first $k$ rows and $k!$
permutations of the first $k$ columns of the determinant defining
$\cK$ and expanding that determinant with respect to the
\mbox{$k$th} column one obtains
 \begin{align}{\label{rek6}}
 \nonumber & \cL(x_{1},\ldots,x_{k},y_{1},\ldots,y_{k})\\
 \nonumber & =\newt{n}{k}\frac{1}{n!}\frac{1}{(k!)^{2}}
 \sum_{\pi,\sigma\in S_{k}} \sgn\,\pi\,\sgn\,\sigma
 \sum_{j=1}^{k}(-1)^{k+j}
 \int_{Y^{n-k}}\varrho(x_{\pi(j)},y_{\sigma(k)}) \\
 \nonumber & \qquad \cdot \varrho^{\wedge(n-1)}
 \left(
 \begin{array}{c}
 x_{\pi(1)},\ldots,x_{\pi(j-1)},x_{\pi(j+1)},
 \ldots,x_{\pi(k)},x_{k+1},\ldots,x_{n}\\
 y_{\sigma(1)},\ldots,y_{\sigma(k-1)},x_{k+1},\ldots,x_{n}
 \end{array}
 \right) \\
 \nonumber & \qquad \dx\mu^{\otimes(n-k)}(x_{k+1},\ldots,x_{n}) \\
 \nonumber & \quad +\newt{n}{k}\frac{1}{n!}\frac{1}{(k!)^{2}}
 \sum_{\pi,\sigma\in S_{k}} \sgn\,\pi\,\sgn\,\sigma
 \sum_{j=k+1}^{n}(-1)^{k+j}
 \int_{Y^{n-k}}\varrho(x_{j},y_{\sigma(k)}) \\
 \nonumber & \qquad \cdot\varrho^{\wedge(n-1)}
 \left(
 \begin{array}{c}
 x_{\pi(1)},\ldots,x_{\pi(k)},x_{k+1},
 \ldots,x_{j-1},x_{j+1},\ldots,x_{n}\\
 y_{\sigma(1)},\ldots,y_{\sigma(k-1)},x_{k+1},\ldots,x_{n}
 \end{array}
 \right) \\
 & \qquad \dx\mu^{\otimes(n-k)}(x_{k+1},\ldots,x_{n}).
 \end{align}

Consider the first term on the r.h.s. of Eq.~\eqref{rek6}. In all
summands of $\sum_{j=1}^{k}$ except the last one the
\mbox{$(k-1)$th} row of the determinant (containing the variable
$x_{\pi(k)}$) can be shifted into the \mbox{$j$th} position,
changing thereby the sign of the determinant by
 $(-1)^{(k-2)-(j-1)} =(-1)^{-k-j+1}.$
Then the first term of sum~\eqref{rek6} assumes the form
 \begin{align}{\label{rek9}}
 \nonumber & \newt{n}{k}\frac{1}{n!}\frac{1}{(k!)^{2}}
  \sum_{\pi,\sigma\in S_{k}}
 \sgn\,\pi\,\sgn\,\sigma\sum_{j=1}^{k-1}(-1)^{k+j}(-1)^{-k-j+1}
 \int_{Y^{n-k}}\varrho(x_{\pi(j)},y_{\sigma(k)}) \\
 \nonumber & \quad \cdot \varrho^{\wedge(n-1)}
 \left(
 \begin{array}{c}
 x_{\pi(1)},\ldots,x_{\pi(j-1)},x_{\pi(k)},x_{\pi(j+1)},
  \ldots,x_{\pi(k-1)},x_{k+1},\ldots,x_{n} \\
 y_{\sigma(1)},\ldots,y_{\sigma(k-1)},x_{k+1},\ldots,x_{n}
 \end{array}\right) \\
 \nonumber & \quad \dx\mu^{\otimes(n-k)}(x_{k+1},\ldots,x_{n}) \\
 \nonumber & +\newt{n}{k}\frac{1}{n!}\frac{1}{(k!)^{2}}
 \sum_{\pi,\sigma\in S_{k}} \sgn\,\pi\,\sgn\,\sigma\,(-1)^{k+k}
 \int_{Y^{n-k}}\varrho(x_{\pi(k)},y_{\sigma(k)})  \\
 & \quad \cdot \varrho^{\wedge(n-1)}
 \left(
 \begin{array}{c}
 x_{\pi(1)},\ldots,x_{\pi(k-1)},x_{k+1},\ldots,x_{n} \\
 y_{\sigma(1)},\ldots,y_{\sigma(k-1)},x_{k+1},\ldots,x_{n}
 \end{array}
 \right)
 \dx\mu^{\otimes(n-k)}(x_{k+1},\ldots,x_{n}).
 \end{align}
Let $T_{jk}\in S_{k}$ denote the transposition $j\leftrightarrow k$
for $j<k$ (then
 $(-1)^{k+j}(-1)^{-k-j+1} =(-1) =\sgn\,T_{jk}$)
and the identity permutation for $j=k$ (with $\sgn\,T_{kk}=1$).
Expression~\eqref{rek9} can be written as
 \begin{align*}
 & \sum_{j=1}^{k} \newt{n}{k}\frac{1}{n!}\frac{1}{(k!)^{2}}
 \sum_{\pi,\sigma\in S_{k}}(\sgn\,\pi\,\sgn\,T_{jk})\,\sgn\,\sigma
 \int_{Y^{n-k}}\varrho(x_{(\pi\circ T_{jk})(k)},y_{\sigma(k)}) \\
 & \quad \cdot \varrho^{\wedge(n-1)}
 \left(
 \begin{array}{c}
 x_{(\pi\circ T_{jk})(1)},
 \ldots,x_{(\pi\circ T_{jk})(k-1)},x_{k+1},\ldots,x_{n} \\
 y_{\sigma(1)},\ldots,y_{\sigma(k-1)},x_{k+1},\ldots,x_{n}
 \end{array}\right) \\
 &\quad \dx\mu^{\otimes(n-k)}(x_{k+1},\ldots,x_{n})
 \end{align*}
 \begin{align}{\label{rek10}}
 \nonumber & =\newt{n-1}{k-1}\frac{1}{(k!)^{2}}
 \sum_{\tau,\sigma\in S_{k}} \sgn\,\tau\,\sgn\,\sigma
 \varrho(x_{\tau(k)},y_{\sigma(k)})\int_{Y^{n-k}}\frac{1}{(n-1)!} \\
 & \quad \cdot \varrho^{\wedge(n-1)}
 \left(
 \begin{array}{c}
 x_{\tau(1)},\ldots,x_{\tau(k-1)},x_{k+1},\ldots,x_{n} \\
 y_{\sigma(1)},\ldots,y_{\sigma(k-1)},x_{k+1},\ldots,x_{n}
 \end{array}
 \right)
 \dx\mu^{\otimes(n-k)}(x_{k+1},\ldots,x_{n}).
 \end{align}
The function $\cP_1\colon Y^k\times Y^k\to\K,$
such that
 $\cP_{1}(x_{1},\ldots,x_{k},y_{1},\ldots,y_{k})$
is \mbox{$\mu^{\otimes 2k}$-a.e.} equal to expression~\eqref{rek10},
is an integral kernel of the operator
 \begin{equation*}
 \newt{n-1}{k-1}\left(\Lc^{k-1}_{n-1}\rho^{\wedge(n-1)}\right)\wedge \rho,
 \end{equation*}
which appears on the r.h.s. of Eq.~\eqref{rek1}.

Consider now the second term of the sum on the r.h.s. of
Eq.~\eqref{rek6}. One can change the indices of the integral
variables $x_{k+1},\ldots,x_{j}$ in all summands of
$\sum_{j=k+1}^{n}$ except the first one, according to the rule
 $x_{j}\to x_{k+1}\to x_{k+2} \to\cdots\to x_{j}$
for the \mbox{$j$th} summand, and simultaneously change the order
of the columns of the determinant inversely (which changes the sign
by
 $(-1)^{(j-1)-k} =(-1)^{(k+1)-j}$).
The resulting sum $\sum_{j=k+1}^{n}$ then contains $n-k$ terms
identical to the one with $j=k+1.$
Thus the second term of sum~\eqref{rek6} equals
 \begin{align*}
 &-(n-k)\newt{n}{k}\frac{1}{n!}\frac{1}{(k!)^{2}}
 \sum_{\pi,\sigma\in S_{k}}
 \sgn\,\pi\,\sgn\,\sigma\int_{Y^{n-k}}
 \varrho(x_{k+1},y_{\sigma(k)}) \\
 &\qquad \cdot \varrho^{\wedge(n-1)}
 \left(
 \begin{array}{c}
 x_{\pi(1)},\ldots,x_{\pi(k)},x_{k+2},\ldots,x_{n} \\
 y_{\sigma(1)},\ldots,y_{\sigma(k-1)},x_{k+1},\ldots,x_{n}
 \end{array}
 \right)
 \dx\mu^{\otimes(n-k)}(x_{k+1},\ldots,x_{n}).
 \end{align*}
 \begin{align}{\label{rek8}}
 &\nonumber =-\newt{n-1}{k}\frac{1}{k!}\sum_{\sigma\in S_{k}}
 \sgn\,\sigma\int_{Y}\varrho(x_{k+1},y_{\sigma(k)})
 \left(\int_{Y^{n-1-k}}\frac{1}{(n-1)!} \right. \\
 &\nonumber \qquad \cdot \varrho^{\wedge(n-1)}
 \left(
 \begin{array}{c}
 x_{1},\ldots,x_{k},x_{k+2},\ldots,x_{n} \\
 y_{\sigma(1)},\ldots,y_{\sigma(k-1)},x_{k+1},\ldots,x_{n}
 \end{array}
 \right) \\
 &\qquad \left. \dx\mu^{\otimes(n-1-k)}(x_{k+2},\ldots,x_{n})\right)
 \dx\mu(x_{k+1}).
 \end{align}
The function $\cP_2\colon Y^k\times Y^k\to\K,$ such that
 $\cP_{2}(x_{1},\ldots,x_{k},y_{1},\ldots,y_{k})$
is \mbox{$\mu^{\otimes 2k}$-a.e.} equal to expression~\eqref{rek8},
is an integral kernel of the operator
 \begin{equation*}
 -\newt{n-1}{k}\left(\Lc^{k}_{n-1}\rho^{\wedge(n-1)}\right)
 \left(\Id^{\otimes(k-1)}\otimes \rho\right)A^{(k)}_{\cH_{Y}},
 \end{equation*}
which occurs on the r.h.s. of Eq.~\eqref{rek1}. One concludes that
the kernel $\cL$ of the operator on the l.h.s. of Eq.~\eqref{rek1}
is \mbox{$\mu^{\otimes 2k}$-a.e.} equal to the kernel
$\cP_{1}+\cP_{2}$ of the operator on the r.h.s.
of Eq.~\eqref{rek1}, which proves the equality of both operators.

The proof of Eq.~\eqref{rek1} for $n=k+1$ and the proof of
Eq.~\eqref{rek3} proceed analogously.

Similarly, the proof of Eqs.~\eqref{rek2}, \eqref{rek4} is
accomplished by changing the product $\wedge$ into $\vee$ and
replacing determinants in all formulae by pernaments, defined for
every complex matrix
 $A=[a_{i,j}]_{i,j=1}^{m}$ as
 \begin{equation*}
 \per A =\sum_{\pi\in S_{m}}a_{\pi(1),1} \cdots a_{\pi(m),m}.
 \end{equation*}
Notice that signs of permutations are omitted in this case,
similarly as the multipliers $\pm 1$ in the Laplace expansions.
\end{proof}

 \begin{lemma}{\label{komkontr}}
 Let $k,m\in\N,$ $1<k<m,$ $\rho\in\cT\left(\cH_{Y}\right),$
  $j_{k}\in\set{k,\ldots,m},$ and
 \begin{equation*}
 R\eqdf\sum_{j_{k-1}=k-1}^{j_{k}-1} \sum_{j_{k-2}=k-2}^{j_{k-1}-1}
 \ldots\sum_{j_{1}=1}^{j_{2}-1}\rho^{j_{1}}
 \otimes \rho^{j_{2}-j_{1}}\otimes\cdots \otimes \rho^{j_{k}-j_{k-1}}
 \end{equation*}
(for $k=2$ the only summation index is $j_{k-1}=j_{1}).$ Then
  $A^{(k)}_{\cH_{Y}}R =RA^{(k)}_{\cH_{Y}}$
 and
  $S^{(k)}_{\cH_{Y}}R =RS^{(k)}_{\cH_{Y}}.$
 \end{lemma}
The proof of the above lemma consists in demonstrating
the invariance of $R$ under permutations of factors
in the tensor products.
To this end it suffices to observe that $R$ is invariant under
transpositions of neighbouring factors.

 \begin{lemma}{\label{pomoc}}
 Let $\rho\in\cT(\cH_{Y}),$
 $\xi^{\wedge}_{s} \eqdf\Tr\,\rho^{\wedge s},$
 $\xi^{\vee}_{s} \eqdf\Tr\,\rho^{\vee s}$
 for $s\in\N,$ $\xi^{\wedge}_{0} \eqdf 1,$
 $\xi^{\vee}_{0} \eqdf 1,$ and
 \begin{equation*}
 \Pi^{\wedge p}_{m}(\rho)
 \eqdf\sum_{j_{p}=p}^{m}\sum_{j_{p-1}=p-1}^{j_{p}-1}
 \ldots\sum_{j_{1}=1}^{j_{2}-1}
 {\xi}^{\wedge}_{m-j_{p}}(-1)^{p+j_{p}}\rho^{j_{1}}
 \wedge \rho^{j_{2}-j_{1}}\wedge\ldots\wedge \rho^{j_{p}-j_{p-1}},
 \end{equation*}
 \begin{equation*}
 \Pi^{\vee p}_{m}(\rho)
 \eqdf\sum_{j_{p}=p}^{m}\sum_{j_{p-1}=p-1}^{j_{p}-1}
 \ldots\sum_{j_{1}=1}^{j_{2}-1} {\xi}^{\vee}_{m-j_{p}}\rho^{j_{1}}
 \vee \rho^{j_{2}-j_{1}}\vee\cdots \vee \rho^{j_{p}-j_{p-1}}
 \end{equation*}
 for $p,m\in\N,$ $p\leq m.$ (For $p=1$ the only summation index is
 $j_{1}$ and the summation runs over the operators $\rho^{j_{1}}.$) If
 $2\leq p<m$ then
 \begin{equation}{\label{pomoc1}}
 \Pi^{\wedge p}_{m}(\rho)
 =\left(\Pi^{\wedge(p-1)}_{m-1}(\rho)\right)\wedge\rho
 -\left(\Pi^{\wedge p}_{m-1}(\rho)\right)
 (\Id^{\otimes(p-1)}\otimes \rho)A^{(p)}_{\cH_{Y}}
 \end{equation}
and
 \begin{equation}{\label{pomoc2}}
 \Pi^{\vee p}_{m}(\rho)
 =\left(\Pi^{\vee (p-1)}_{m-1}(\rho)\right)\vee \rho
 +\left(\Pi^{\vee p}_{m-1}(\rho)\right)
 (\Id^{\otimes(p-1)}\otimes \rho)S^{(p)}_{\cH_{Y}}.
 \end{equation}
 \end{lemma}
\begin{proof}
Eq.~\eqref{pomoc1} will be first proved for $p>2.$  One has
 \begin{align*}
 & \Pi^{\wedge p}_{m}(\rho)
 =\xi^{\wedge}_{m-p}\rho\wedge\ldots\wedge \rho
 +\sum_{l_{p}=p}^{m-1}\sum_{l_{p-1}=p-1}^{l_{p}}
 \sum_{l_{p-2}=p-2}^{l_{p-1}-1}\ldots \sum_{l_{1}=1}^{l_{2}-1} \\
 &\qquad {\xi}^{\wedge}_{m-l_{p}-1}(-1)^{p+l_{p}+1}\rho^{l_{1}}
 \wedge \rho^{l_{2}-l_{1}}\wedge\ldots\wedge \rho^{l_{p-1}-l_{p-2}}
 \wedge \rho^{l_{p}-l_{p-1}+1}
 \end{align*}
 \begin{align}{\label{pomoc4}}
 \nonumber & =\xi^{\wedge}_{m-p}\rho\wedge\ldots\wedge \rho
 +\sum_{l_{p}=p}^{m-1}\sum_{l_{p-1}=p-1}^{l_{p}-1}
 \sum_{l_{p-2}=p-2}^{l_{p-1}-1}\ldots \sum_{l_{1}=1}^{l_{2}-1} \\
 &\nonumber \qquad {\xi}^{\wedge}_{m-l_{p}-1}
 (-1)^{p+l_{p}+1}\rho^{l_{1}}
 \wedge \rho^{l_{2}-l_{1}}\wedge\ldots\wedge \rho^{l_{p-1}-l_{p-2}}
 \wedge \rho^{l_{p}-l_{p-1}+1} \\
 \nonumber &\quad +\sum_{l_{p}=p}^{m-1}
 \sum_{l_{p-2}=p-2}^{l_{p}-1}\sum_{l_{p-3}=p-3}^{l_{p-2}-1}
 \ldots \sum_{l_{1}=1}^{l_{2}-1} \\
 &\qquad {\xi}^{\wedge}_{m-l_{p}-1}(-1)^{p+l_{p}+1}\rho^{l_{1}}
 \wedge \rho^{l_{2}-l_{1}}\wedge\ldots\wedge \rho^{l_{p-2}-l_{p-3}}
 \wedge \rho^{l_{p}-l_{p-2}}\wedge \rho.
 \end{align}
The first and the third term after the last
of equalities~\eqref{pomoc4} yield
 \begin{align}{\label{pomoc5}}
 \nonumber & \sum_{j_{p-1}=p-1}^{m-1}
 \sum_{j_{p-2}=p-2}^{j_{p-1}-1}\ldots \sum_{j_{1}=1}^{j_{2}-1} \\
 &\quad \left({\xi}^{\wedge}_{(m-1)-j_{p-1}}
 (-1)^{(p-1)+j_{p-1}}\rho^{j_{1}}
 \wedge \rho^{j_{2}-j_{1}}\wedge
 \ldots\wedge \rho^{j_{p-1}-j_{p-2}}\right)\wedge \rho
 \end{align}
for $l_{p}\equiv j_{p-1},$
$l_{p-2}\equiv j_{p-2},\ldots,l_{1}\equiv j_{1}.$
By Lemma~\ref{komkontr}, the second term after the last
of equalities~\eqref{pomoc4} equals
 \begin{equation}{\label{pomoc6}}
 -\left(\Pi^{\wedge p}_{m-1}(\rho)\right)
 (\Id^{\otimes(p-1)}\otimes \rho)A^{(p)}_{\cH_{Y}}.
 \end{equation}
The sum of expressions~\eqref{pomoc5} and~\eqref{pomoc6}
is equal to the r.h.s. of Eq.~\eqref{pomoc1} for $p>2.$
After simplifications the proof also applies to the case of $p=2.$

The proof of Eq.~\eqref{pomoc2} is analogous to that of
Eq.~\eqref{pomoc1}.
\end{proof}

The next theorem provides the explicit form of
\mbox{$(n,k)$-con}tractions of product operators. The proof for
$k=1,2$ was given in~\cite{KossakowskiRMP86,MackowiakPR99}.
The author of~\cite{MackowiakPR99} emphasized that
formula~\eqref{jawny1th} for $k=2$ was derived by S. Pruski in 1978.

 \begin{theorem}[Explicit formulae]{\label{jawny}}
 Let $k,n\in\N,$ $k<n,$ $\rho\in\cT\left(\cH_{Y}\right),$
 $\xi^{\wedge}_{s} \eqdf\Tr\,\rho^{\wedge s},$
  $\xi^{\vee}_{s} \eqdf\Tr\,\rho^{\vee s}$
 for $s\in\N,$ and $\xi^{\wedge}_{0} \eqdf 1,$
 $\xi^{\vee}_{0} \eqdf 1.$ Then
 \begin{align}{\label{jawny1th}}
 &\nonumber  \newt{n}{k}\Lc^{k}_{n}\rho^{\wedge n} \\
 &\nonumber =\sum_{j_{k}=k}^{n}\sum_{j_{k-1}=k-1}^{j_{k}-1}
 \ldots\sum_{j_{1}=1}^{j_{2}-1}
 {\xi}^{\wedge}_{n-j_{k}}(-1)^{k+j_{k}}\rho^{j_{1}}
 \wedge \rho^{j_{2}-j_{1}}\wedge
 \ldots \wedge \rho^{j_{k}-j_{k-1}} \\
 \nonumber & =\sum_{i_{1}=1}^{n-(k-1)}
 \sum_{i_{2}=1}^{n-i_{1}-(k-2)}\ldots
 \sum_{i_{k-1}=1}^{n-i_{1}-\cdots-i_{k-2}-1}
 \sum_{i_{k}=1}^{n-i_{1}-\cdots-i_{k-1}} \\
 &\quad {\xi}^{\wedge}_{n-i_{1}-\cdots-i_{k}}
 (-1)^{k+i_{1}+\cdots+i_{k}}\rho^{i_{1}}\wedge\ldots\wedge \rho^{i_{k}}
 \end{align}
 and
 \begin{align}{\label{jawny2th}}
 \nonumber & \newt{n}{k}\Lc^{k}_{n}\rho^{\vee n} \\
 \nonumber & =\sum_{j_{k}=k}^{n}
 \sum_{j_{k-1}=k-1}^{j_{k}-1} \ldots\sum_{j_{1}=1}^{j_{2}-1}
 {\xi}^{\vee}_{n-j_{k}}\rho^{j_{1}}\vee \rho^{j_{2}-j_{1}}\vee
 \cdots \vee \rho^{j_{k}-j_{k-1}} \\
 &\nonumber =\sum_{i_{1}=1}^{n-(k-1)}\sum_{i_{2}=1}^{n-i_{1}-(k-2)}
 \ldots \sum_{i_{k-1}=1}^{n-i_{1}-\cdots-i_{k-2}-1}
 \sum_{i_{k}=1}^{n-i_{1}-\cdots-i_{k-1}} \\
 &\quad {\xi}^{\vee}_{n-i_{1}-\cdots-i_{k}}
 \rho^{i_{1}}\vee\cdots\vee \rho^{i_{k}}.
 \end{align}
 (For $k=1$ the only summation indices are $j_1$ and $i_1$ and the
 summation runs over the operators $\rho^{j_1}$ and $\rho^{i_1}$,
 respectively.)
 \end{theorem}
\begin{proof} For every $p,m\in\N,$ $p\leq m,$ let
 $\Pi^{\wedge p}_{m}(\rho)$
be defined as in Lemma~\ref{pomoc}. Then the first of
equalities~\eqref{jawny1th} can be written as
 \begin{equation}{\label{jawny1}}
 \newt{n}{k}\Lc^{k}_{n}\rho^{\wedge n}
 =\Pi^{\wedge k}_{n}(\rho).
 \end{equation}
The proof of Eq.~\eqref{jawny1} will be carried out by (double)
induction with respect to $k$ and, for fixed $k,$ with respect to
$n>k.$
\newline
$1^{\circ}.$ ($k=1$) This part of the proof is by
induction with respect to $n>1.$
\newline
a) ($n=2$) According to Theorem~\ref{rek},
 $2\Lc^{1}_{2}\rho^{\wedge 2} =\left(\Tr\,\rho\right)\rho-\rho^{2}
 =\Pi^{\wedge 1}_{2}(\rho).$
\newline
b) Assuming validity of formula~\eqref{jawny1} (with $k=1)$ for
 $n\in\set{2,\ldots,m-1},$
where $m\in\N,$ $m>2,$ its validity will be shown for $n=m.$

One has
 \begin{equation*}
 \Pi^{\wedge 1}_{m}(\rho)
 =\xi^{\wedge}_{m-1}\rho
 +\sum_{j_{1}=2}^{m}\xi^{\wedge}_{m-j_{1}}
 (-1)^{1+j_{1}}\rho^{j_{1}}
 =\xi^{\wedge}_{m-1}\rho
 -\left(\Pi^{\wedge 1}_{m-1}(\rho)\right)\rho.
 \end{equation*}
Thus, according to the inductive hypothesis for
$n\in\set{2,\ldots,m-1},$
 \begin{equation*}
 \Pi^{\wedge 1}_{m}(\rho)
 =\xi^{\wedge}_{m-1}\rho
 -(m-1)\left(\Lc^{1}_{m-1}\rho^{\wedge(m-1)}\right)\rho,
 \end{equation*}
which, in view of Theorem~\ref{rek}, yields
 $\newt{m}{1}\Lc^{1}_{m}\rho^{\wedge m}
 =\Pi^{\wedge 1}_{m}(\rho).$
\newline
$2^{\circ}.$ Assuming validity of formula~\eqref{jawny1} for
$k\in\set{1,\ldots,p-1}$ (and every $n>k),$ where $p\in\N,$ $p>1,$
its validity will be shown for $k=p.$ For arbitrarily fixed $p$ the
proof will be carried out by induction with respect to $n>p.$
\newline
a) ($n=p+1$) By the inductive hypothesis with respect to $k$ and
Lemma~\ref{pomoc},
 \begin{equation*}
 \Pi^{\wedge p}_{p+1}(\rho)
 =\newt{(p+1)-1}{p-1}
 \left(\Lc^{p-1}_{(p+1)-1}\rho^{\wedge((p+1)-1)}\right)\wedge \rho
 -\rho^{\wedge p}(\Id^{\otimes(p-1)}\otimes \rho)A^{(p)}_{\cH_{Y}},
 \end{equation*}
hence
 $\displaystyle \newt{p+1}{p}\Lc^{p}_{p+1}\rho^{\wedge (p+1)}
 =\Pi^{\wedge p}_{p+1}(\rho),$
according to Theorem~\ref{rek}.
\newline
b) Assuming validity of formula~\eqref{jawny1} for
$n\in\set{p+1,\ldots,m-1},$ where $k=p,$ $m\in\N,$ $m>p+1,$ its
validity will be shown for $n=m.$

By the inductive hypothesis for $k\in\set{1,\ldots,p-1}$
and Lemma~\ref{pomoc} one has
 \begin{equation*}
 \Pi^{\wedge p}_{m}(\rho)
 =\newt{m-1}{p-1}
 \left(\Lc^{p-1}_{m-1}\rho^{\wedge(m-1)}\right)\wedge \rho
 -\left(\Pi^{\wedge p}_{m-1}(\rho)\right)
 (\Id^{\otimes(p-1)}\otimes \rho)A^{(p)}_{\cH_{Y}}.
 \end{equation*}
According to the inductive hypothesis for
$n\in\set{p+1,\ldots,m-1}$ one thus obtains
 \begin{multline*}
 \Pi^{\wedge p}_{m}(\rho)
 =\newt{m-1}{p-1}\left(\Lc^{p-1}_{m-1}\rho^{\wedge(m-1)}\right)
 \wedge \rho \\
 -\newt{m-1}{p}\left(\Lc^{p}_{m-1}\rho^{\wedge(m-1)}\right)
 \left(\Id^{\otimes(p-1)}\otimes \rho\right)A^{(p)}_{\cH_{Y}},
 \end{multline*}
which, in view of Theorem~\ref{rek}, yields
 $\newt{m}{p}\Lc^{p}_{m}\rho^{\wedge m}
 =\Pi^{\wedge p}_{m}(\rho).$
This completes the inductive proof for Eq.~\eqref{jawny1} with
respect to $n>p$ and with respect to $k.$

Now turn to the second of equalities~\eqref{jawny1th}.
For $k=1$ it is identity. Let $k\geq 2.$ Setting
 $j_{1}=i_{1},$ $j_{2}=i_{1}+i_{2},$
 \ldots, $j_{k}=i_{1}+\cdots+i_{k}$
or, equivalently,
 $i_{1}=j_{1},$ $i_{2}=j_{2}-j_{1},$ $i_{3}=j_{3}-j_{2},$
 \ldots, $i_{k}=j_{k}-j_{k-1},$
one checks that both sides of the equality in question are
equal to
 \begin{multline*}
 \sum_{j_{1}=1}^{n-(k-1)}\sum_{j_{2}=j_{1}+1}^{n-(k-2)} \ldots
 \sum_{j_{k-1}=j_{k-2}+1}^{n-1}\sum_{j_{k}=j_{k-1}+1}^{n} \\
 {\xi}^{\wedge}_{n-j_{k}}(-1)^{k+j_{k}}\rho^{j_{1}}
 \wedge \rho^{j_{2}-j_{1}}\wedge\ldots \wedge \rho^{j_{k}-j_{k-1}}.
 \end{multline*}

The proof of Eq.~\eqref{jawny2th} is analogous to that of
Eq.~\eqref{jawny1th}.
\end{proof}

\section{Asymptotic form for contractions of product
states}{\label{rozas}}

The explicit forms of the contractions of product states given by
Theorem~\ref{jawny} are quite complex. In the present section they
are replaced by simpler operators, equivalent in the thermodynamic
limit. The main results in this section are Theorems~\ref{glowne}
and~\ref{zmiana}.

In what follows use is made of the Hilbert space
 $\cH_{\Omega}\eqdf L^{2}(\Omega,\mu)$
(over $\C$ or $\R$), where the measure $\mu$ is separable,
\mbox{$\sigma$-fi}nite, and satisfies the condition
$\mu(\Omega)=+\infty.$ For every \mbox{$\mu$-measu}rable subset
$Y\subset\Omega$ it is assumed $\cH_{Y}\eqdf L^{2}(Y,\mu).$

Let $\cM(\Omega)$ be a~fixed family of measurable subsets of
$\Omega$ such that $0<\mu(Y)<+\infty$ for every $Y\in\cM(\Omega)$
(it can be the family of all such subsets). Fix $d\in\R,$ $d>0,$
and assume that there exists a~sequence
 $\set{Y_{n}}_{n\in\N} \subset\cM(\Omega)$
such that $\frac{n}{\mu(Y_{n})}\to d$ as $n\to\infty.$

 \begin{definition}{\label{granica}} Fix $d\in\R,$ $d>0,$ and let
  $\set{b_{Y,n}}_{(Y,n)\in\cM(\Omega)\times\N}$
 be a~family of complex numbers. A~complex number $b$ is said to be
 the \emph{thermodynamic limit} of this family if for every sequence
  $\set{Y_{n}}_{n\in\N}\subset\cM(\Omega)$
 such that
  $\displaystyle\lim_{n\to\infty}\frac{n}{\mu(Y_{n})}=d$
 the condition
  $\displaystyle \lim_{n\to\infty}b_{Y_{n},n}=b$
 is fulfilled. In such a~case $b$ is denoted by
  $\displaystyle\limt{n}{Y}b_{Y,n}.$
 \end{definition}

Special attention will be given to the families of
complex numbers of the form
 $\Tr\,(\Lc^{k}_{n}K_{Y,n})C_{Y},$
where $k,n\in\N,$ $n>k,$
 $K_{Y,n}\in\cT(\cH_{Y}^{\otimes n}),$
and
 $C_{Y}\in\cB(\cH_{Y}^{\otimes k}).$

Definition~\ref{granica} does not guarantee the convergence
of families $\set{b_{Y,n}}$ of interest in physics. To obtain such
a~convergence, additional conditions (such as conditions of uniform
growth~\cite{RuelleB69}) are usually imposed  on the sequence
$\set{Y_n}_{n\in\N}$ in question. However, those additional
conditions do not affect considerations in this paper.

Expression of expectation values of observables in mixed states
by using trace, mentioned in Introduction, is the motivation for the
following definition.

 \begin{definition}{\label{relacja}}
 Fix $k\in\N$ and $d\in\R,$ $d>0.$ Families $\rodz{A}{Y}{n}$
 and $\rodz{B}{Y}{n}$ of operators
  $A_{Y,n},B_{Y,n}\in\cT(\cH_{Y}^{\otimes k})$
 are said to be \emph{asymptotically equivalent} (symbolically:
  $A_{Y,n}\approx B_{Y,n}),$
 if for every family $\rodz{C}{Y}{n}$ of operators
  $C_{Y,n}\in\cB(\cH_{Y}^{\otimes k})$
 with uniformly bounded operator norms one has
 \begin{equation}{\label{relacja1}}
 \limt{n}{Y}\Tr\,(A_{Y,n}-B_{Y,n})C_{Y,n}=0.
 \end{equation}
 \end{definition}

Condition~\eqref{relacja1} is required to hold
in particular for families $\rodz{C}{Y}{n}$ such that
$C_{Y,n}=C_{Y,m}$
for all $Y\in\cM(\Omega),$ $n,m\in\N.$

 \begin{remark}{\label{defasdiff}}
 The authors of~\cite{KossakowskiRMP86,MackowiakPR99}
 used some different definition
 of asymptotic equivalence of families of operators, closer to
 Definition~\ref{rel} in this paper.
 \end{remark}

 \begin{remark}{\label{strongprop}}
 For fixed $k\in\N$ and $d\in\R,$ $d>0,$ the relation
 $\approx$ is an equivalence relation. If $A_{Y,n}\approx B_{Y,n}$
 then for every family of operators $C_{Y,n}$ as in
 Definition~\ref{relacja} the limit
 $\displaystyle\limt{n}{Y}\Tr\,A_{Y,n}C_{Y,n}$ exists iff the limit
 $\displaystyle\limt{n}{Y}\Tr\,B_{Y,n}C_{Y,n}$ exists, in which
 case both limits are equal. Notice also that if
  $A_{Y,n}\approx B_{Y,n}$
 then
  $A_{Y,n}+D_{Y,n} \approx B_{Y,n}+D_{Y,n}$
 and
  $aA_{Y,n}\approx aB_{Y,n}$
 for every family
  $\rodz{D}{Y}{n}\subset \cT(\cH_{Y}^{\otimes k})$
 and $a\in\C.$ Furthermore, for every family
  $\rodz{A}{Y}{n}\subset \cT\left(\cH_{Y}^{\otimes k}\right)$
 with uniformly bounded trace norms $\Tr\,\abs{A_{Y,n}}$ and for
 every sequence
  $\set{a_{n}}_{n\in\N}\subset\C$
 convergent to $a\in\C$ one has
  $a_{n}A_{Y,n}\approx aA_{Y,n}.$
 \end{remark}

 \begin{lemma}{\label{mocna}}
 Let $\rodz{A}{Y}{n}$ and $\rodz{B}{Y}{n}$ be as in
 Definition~\ref{relacja}. Then
 \begin{equation}{\label{mocna1}}
 \limt{n}{Y}\abs{A_{Y,n}-B_{Y,n}}=0
 \quad
 \Rightarrow \quad A_{Y,n}\approx B_{Y,n}.
 \end{equation}
 Moreover, if the operators $A_{Y,n},$ $B_{Y,n}$ are selfadjoint
 then
 \begin{equation}{\label{mocna2}}
 A_{Y,n} \approx B_{Y,n} \quad \Rightarrow
 \quad
 \limt{n}{Y}\abs{A_{Y,n}-B_{Y,n}}=0.
 \end{equation}
 \end{lemma}
\begin{proof}
Implication \eqref{mocna1} follows from Definition~\ref{relacja}
and the estimate
 \begin{equation*}
 \abs{\Tr\,(A_{Y,n}-B_{Y,n})C_{Y,n}}
 \leq\norm{C_{Y,n}}\,\Tr\,\abs{A_{Y,n}-B_{Y,n}}.
 \end{equation*}

Now assume that $A_{Y,n}\approx B_{Y,n},$ which is equivalent
to the condition
 \begin{equation}{\label{mocna5}}
 D_{Y,n}\approx 0,
 \end{equation}
where
 $D_{Y,n} \eqdf A_{Y,n}-B_{Y,n}.$
The operators $D_{Y,n}$ have the spectral representations
 \begin{equation*}
 D_{Y,n} =\sum_{i=1}^{\infty}\lambda_{i}(Y,n)P_{\varphi_{i}(Y,n)},
 \end{equation*}
where $P_{\varphi_{i}(Y,n)}$ are the projectors onto orthogonal one
dimensional subspaces of eigenvectors $\varphi_{i}(Y,n)$ of
$D_{Y,n},$ corresponding to eigenvalues
 $\lambda_{i}(Y,n)\in\R.$
Since
 $\displaystyle \sum_{i=1}^{\infty}\abs{\lambda_{i}(Y,n)}
 =\Tr\,\abs{D_{Y,n}}<+\infty,$
for every
 $(Y,n)\in\cM(\Omega)\times\N$
there exists $m(Y,n)\in\N$ such that
 $\displaystyle\sum_{i=m(Y,n)+1}^{\infty}
 \abs{\lambda_{i}(Y,n)}<\frac{1}{n}.$
Thus the operators
 \begin{equation*}
 F_{Y,n} =\sum_{i=1}^{m(Y,n)}\lambda_{i}(Y,n)P_{\varphi_{i}(Y,n)}
 \end{equation*}
satisfy the condition
 \begin{equation}{\label{mocna7}}
 \limt{n}{Y}\Tr\,\abs{D_{Y,n}-F_{Y,n}}
 =\limt{n}{Y}\sum_{i=m(Y,n)+1}^{\infty}\abs{\lambda_{i}(Y,n)}=0,
 \end{equation}
which, in view of implication~\eqref{mocna1} proved and
condition~\eqref{mocna5}, yields $F_{Y,n}\approx D_{Y,n}\approx 0.$
In particular,
 \begin{equation}{\label{mocna6}}
 \limt{n}{Y}\Tr\,F_{Y,n}C_{Y,n}=0,
 \end{equation}
where
 \begin{equation*}
 C_{Y,n} =\sum_{i=1}^{m(Y,n)}
 \sgn\left(\lambda_{i}(Y,n)\right)P_{\varphi_{i}(Y,n)},
 \qquad \norm{C_{Y,n}}=1.
 \end{equation*}
Observe that
 $\Tr\,F_{Y,n}C_{Y,n}
 =\Tr\,\abs{F_{Y,n}},$
hence condition~\eqref{mocna6} gives
 \begin{equation}{\label{mocnad}}
 \limt{n}{Y}\Tr\,\abs{F_{Y,n}}=0.
 \end{equation}
Since
 $\Tr\,\abs{D_{Y,n}}
 \leq\Tr\,\abs{D_{Y,n}-F_{Y,n}} +\Tr\,\abs{F_{Y,n}},$
conditions~\eqref{mocna7} and~\eqref{mocnad} yield
 \begin{equation*}
 \limt{n}{Y}\Tr\,\abs{A_{Y,n}-B_{Y,n}}
 \equiv\limt{n}{Y}\Tr\,\abs{D_{Y,n}}=0,
 \end{equation*}
which proves implication~\eqref{mocna2}.
\end{proof}

The following lemma follows from Lemma~\ref{mocna}.

 \begin{lemma}{\label{iloczyny}}
 Fix $k,m\in\N.$ Let $\rodz{A}{Y}{n}$ and $\rodz{B}{Y}{n}$ be
 families of selfadjoint operators
 $A_{Y,n},B_{Y,n}\in\cT(\cH_{Y}^{\otimes k})$
such that
 $A_{Y,n}\approx B_{Y,n},$
and let $\rodz{D}{Y}{n}$ be a~family of operators
 $D_{Y,n}\in\cT(\cH_{Y}^{\otimes m})$ with uniformly bounded trace
 norms $\Tr\,\abs{D_{Y,n}}.$ Then
 \begin{equation*}
 A_{Y,n}\otimes D_{Y,n}\approx B_{Y,n}\otimes D_{Y,n},
 \quad
 D_{Y,n}\otimes A_{Y,n}\approx D_{Y,n}\otimes B_{Y,n},
 \end{equation*}
 \begin{equation*}
 A_{Y,n}\wedge D_{Y,n}\approx B_{Y,n}\wedge D_{Y,n},
 \quad
 D_{Y,n}\wedge A_{Y,n}\approx D_{Y,n}\wedge B_{Y,n},
 \end{equation*}
 \begin{equation*}
 A_{Y,n}\vee D_{Y,n}\approx B_{Y,n}\vee D_{Y,n},
 \quad
 D_{Y,n}\vee A_{Y,n}\approx D_{Y,n}\vee B_{Y,n}.
 \end{equation*}
 \end{lemma}

In the sequel $\grup{\rho}{Y}$ denotes a~family of
nonnegative definite selfadjoint operators
 $\rho_{Y}\in\cT\left(\cH_{Y}\right),$ and for every
 $(Y,n)\in\cM(\Omega)\times\N$
it is assumed that
 \begin{equation*}
 \xi^{\wedge}_{Y,0} \eqdf 1, \quad  \xi^{\vee}_{Y,0} \eqdf 1,
 \qquad \rho_{Y}^{\wedge 1} \eqdf \rho_{Y},
 \quad \rho_{Y}^{\vee 1} \eqdf \rho_{Y},
 \end{equation*}
 \begin{equation*}
 \xi^{\wedge}_{Y,n} \eqdf  \Tr\,\rho_{Y}^{\wedge n}>0,
 \qquad
 \xi^{\vee}_{Y,n} \eqdf  \Tr\,\rho_{Y}^{\vee n}>0,
 \end{equation*}
 \begin{equation*}
 s^{\wedge}_{Y,n}
 \eqdf \frac{\xi^{\wedge}_{Y,n-1}}{\xi^{\wedge}_{Y,n}},
 \qquad
 s^{\vee}_{Y,n} \eqdf \frac{\xi^{\vee}_{Y,n-1}}{\xi^{\vee}_{Y,n}}.
 \end{equation*}

The objective of this section is to find density operators of the
most simple form which are asymptotically equivalent to the
operators
 \begin{equation*}
 \stackrel{\wedge}{\sigma}^{(k)}_{Y,n} \eqdf\Lc^{k}_{n}
 \left(\frac{1}{\xi^{\wedge}_{Y,n}}\rho_{Y}^{\wedge n}\right),
 \quad
 \stackrel{\vee}{\sigma}^{(k)}_{Y,n} \eqdf\Lc^{k}_{n}
 \left(\frac{1}{\xi^{\vee}_{Y,n}}\rho_{Y}^{\vee n}\right),
 \end{equation*}
defined for fixed $k\in\N$ and every
 $(Y,n)\in\cM(\Omega)\times\N,$ $n>k.$

 \begin{remark}{\label{odwr}}
 For every $(Y,n)\in\cM(\Omega)\times\N$ the operator
  $\Id+s^{\wedge}_{Y,n+1}\rho_{Y}$
 is invertible and
  $\norm{(\Id+s^{\wedge}_{Y,n+1}\rho_{Y})^{-1}}=1.$
 Furthermore, if
  $s^{\vee}_{Y,n+1}\norm{\rho_{Y}}<1$
 then
  $\Id-s^{\vee}_{Y,n+1}\rho_{Y}$
 is invertible and
  $\displaystyle \norm{(\Id-s^{\vee}_{Y,n+1}\rho_{Y})^{-1}}
  =(1-s^{\vee}_{Y,n+1}\norm{\rho_{Y}})^{-1}.$
 \end{remark}

The next theorem is a~version of a~theorem studied
in~\cite{KossakowskiRMP86,MackowiakPR99}
(see Remark~\ref{defasdiff}).

 \begin{theorem}{\label{as}}
 If
  $\stackrel{\wedge}{\sigma}^{(1)}_{Y,n}
  \approx \stackrel{\wedge}{\sigma}^{(1)}_{Y,n+1}$
 and the reals
  $s^{\wedge}_{Y,n+1}\norm{\rho_{Y}},$
  $(Y,n)\in\cM(\Omega)\times\N,$
 are uniformly bounded then
 \begin{equation}{\label{as1}}
 \stackrel{\wedge}{\sigma}^{(1)}_{Y,n}
 (\Id+s^{\wedge}_{Y,n+1}\rho_{Y})
 \approx (n+1)^{-1}s^{\wedge}_{Y,n+1}\rho_{Y},
 \end{equation}
 \begin{equation}{\label{as2}}
 \stackrel{\wedge}{\sigma}^{(1)}_{Y,n}
 \approx (n+1)^{-1}s^{\wedge}_{Y,n+1}\rho_{Y}
 (\Id+s^{\wedge}_{Y,n+1}\rho_{Y})^{-1}.
 \end{equation}
 If
  $\stackrel{\vee}{\sigma}^{(1)}_{Y,n}
  \approx \stackrel{\vee}{\sigma}^{(1)}_{Y,n+1}$
 and the reals
  $s^{\vee}_{Y,n}\norm{\rho_{Y}},$
  $(Y,n)\in\cM(\Omega)\times\N,$
 are uniformly bounded then
 \begin{equation}{\label{as3}}
 \stackrel{\vee}{\sigma}^{(1)}_{Y,n}(\Id-s^{\vee}_{Y,n+1}\rho_{Y})
 \approx (n+1)^{-1}s^{\vee}_{Y,n+1}\rho_{Y}.
 \end{equation}
 If, additionally,
  $s^{\vee}_{Y,n}\norm{\rho_{Y}}\leq\epsilon$
 for some $\epsilon<1$ and every
  $(Y,n)\in\cM(\Omega)\times\N$ then
 \begin{equation}{\label{as4}}
 \stackrel{\vee}{\sigma}^{(1)}_{Y,n}
 \approx (n+1)^{-1}s^{\vee}_{Y,n+1}\rho_{Y}
 (\Id-s^{\vee}_{Y,n+1}\rho_{Y})^{-1}.
 \end{equation}
 \end{theorem}
\begin{proof}
By Theorem~\ref{rek} and the assumption
 $\stackrel{\wedge}{\sigma}^{(1)}_{Y,n}
 \approx \stackrel{\wedge}{\sigma}^{(1)}_{Y,n+1}$
one has
 \begin{equation}{\label{strongprel}}
 \stackrel{\wedge}{\sigma}^{(1)}_{Y,n}
 -(n+1)^{-1}s^{\wedge}_{Y,n+1}\rho_{Y}
 \approx -(n+1)^{-1}n\stackrel{\wedge}{\sigma}^{(1)}_{Y,n}
 (s^{\wedge}_{Y,n+1}\rho_{Y}).
 \end{equation}
Since
 $\Tr\,\abs{\stackrel{\wedge}{\sigma}^{(1)}_{Y,n}
 (s^{\wedge}_{Y,n+1}\rho_{Y})}
 \leq s^{\wedge}_{Y,n+1}\norm{\rho_{Y}}\,
 \Tr\,\abs{\stackrel{\wedge}{\sigma}^{(1)}_{Y,n}}
 =s^{\wedge}_{Y,n+1}\norm{\rho_{Y}},$
relation~\eqref{strongprel} yields~\eqref{as1},
in view of Remark~\ref{strongprop}.

Now turn to the proof of relation~\eqref{as2}. According to
Remark~\ref{odwr},
 \begin{align}{\label{as5}}
 \nonumber & \Tr\,\abs{\stackrel{\wedge}{\sigma}^{(1)}_{Y,n}
 -(n+1)^{-1}s^{\wedge}_{Y,n+1}\rho_{Y}
 (\Id+s^{\wedge}_{Y,n+1}\rho_{Y})^{-1}} \\
 \nonumber &\leq \norm{(\Id+s^{\wedge}_{Y,n+1}\rho_{Y})^{-1}}
 \,\Tr\,\abs{\stackrel{\wedge}{\sigma}^{(1)}_{Y,n}
 (\Id+s^{\wedge}_{Y,n+1}\rho_{Y})
 -(n+1)^{-1}s^{\wedge}_{Y,n+1}\rho_{Y}} \\
 & =\Tr\,\abs{\stackrel{\wedge}{\sigma}^{(1)}_{Y,n}
 (\Id+s^{\wedge}_{Y,n+1}\rho_{Y})
 -(n+1)^{-1}s^{\wedge}_{Y,n+1}\rho_{Y}}.
 \end{align}
The explicit form of $\stackrel{\wedge}{\sigma}^{(1)}_{Y,n}$ given
by Theorem~\ref{jawny} shows that
 $\stackrel{\wedge}{\sigma}^{(1)}_{Y,n}$
commutes with
 $\Id+s^{\wedge}_{Y,n+1}\rho_{Y},$
and since both operators are selfadjoint,
 $\stackrel{\wedge}{\sigma}^{(1)}_{Y,n}
 (\Id+s^{\wedge}_{Y,n+1}\rho_{Y})$
is also selfadjoint. Thus conditions~\eqref{as1},~\eqref{as5},
and Lemma~\ref{mocna} yield~\eqref{as2}.

The proof of relations~\eqref{as3},~\eqref{as4} runs parallel to
that of~\eqref{as1},~\eqref{as2}.
Notice that in this case the expression
 $\norm{(\Id+s^{\wedge}_{Y,n+1}\rho_{Y})^{-1}}=1$
from estimate~\eqref{as5} is replaced by
 $\norm{(\Id-s^{\wedge}_{Y,n+1}\rho_{Y})^{-1}}
 =(1-s^{\wedge}_{Y,n+1}\norm{\rho_{Y}})^{-1} \leq(1-\epsilon)^{-1}$
(see Remark~\ref{odwr}).
\end{proof}

The following theorem for $k=2$ (with the reservation of
Remark~\ref{defasdiff}) was obtained
in~\cite{KossakowskiRMP86,MackowiakPR99}. The author
of~\cite{MackowiakPR99} gave also arguments that can be used to
check the assumptions of this theorem.

 \begin{theorem}[Asymptotic formulae I]{\label{glowne}}
 If
  $\stackrel{\wedge}{\sigma}^{(k)}_{Y,n}
  \approx \stackrel{\wedge}{\sigma}^{(k)}_{Y,n+1}$
 for every $k\in\N$ and
 \begin{equation}{\label{assferm}}
 s^{\wedge}_{Y,n}\norm{\rho_{Y}} \leq 2
 \quad \text{for every $(Y,n)\in\cM(\Omega)\times\N$}
 \end{equation}
then, for every $k\in\N,$ $k\geq 2,$
 \begin{equation}{\label{glowne2}}
 \stackrel{\wedge}{\sigma}^{(k)}_{Y,n}
 \approx k!\underbrace{\stackrel{\wedge}{\sigma}^{(1)}_{Y,n}\wedge
 \ldots\wedge\stackrel{\wedge}{\sigma}^{(1)}_{Y,n}}_{k}.
 \end{equation}

 If
  $\stackrel{\vee}{\sigma}^{(k)}_{Y,n}
  \approx \stackrel{\vee}{\sigma}^{(k)}_{Y,n+1}$
 for every $k\in\N$ and
 \begin{equation}{\label{assbos}}
 s^{\vee}_{Y,n}\norm{\rho_{Y}} \leq\epsilon
 \quad \text{for some $\epsilon<1$ and every
 $(Y,n)\in\cM(\Omega)\times\N$}
 \end{equation}
 then, for every $k\in\N,$ $k\geq 2,$
 \begin{equation}{\label{glowne4}}
 \stackrel{\vee}{\sigma}^{(k)}_{Y,n}
 \approx k!\underbrace{\stackrel{\vee}{\sigma}^{(1)}_{Y,n}\vee
 \cdots \vee\stackrel{\vee}{\sigma}^{(1)}_{Y,n}}_{k}.
 \end{equation}
 \end{theorem}
\begin{proof}
First equivalence~\eqref{glowne2} will be proved.
Observe that
 \begin{align*}
 & 2\,\Tr\abs{\stackrel{\wedge}{\sigma}^{(q)}_{Y,n}
 -q!\stackrel{\wedge}{\sigma}^{(1)}_{Y,n}\wedge
 \ldots \wedge\stackrel{\wedge}{\sigma}^{(1)}_{Y,n}} \\
 & =\Tr\abs{\left(\stackrel{\wedge}{\sigma}^{(q)}_{Y,n}
 -q!\stackrel{\wedge}{\sigma}^{(1)}_{Y,n}\wedge
 \ldots \wedge\stackrel{\wedge}{\sigma}^{(1)}_{Y,n}\right)
 \left(\Id^{\otimes(q-1)}
 \otimes (\Id+s^{\wedge}_{Y,n+1}\rho_{Y})\right)
 A^{(q)}_{\cH_{Y}} \right. \\
 &\quad \left. +\left(\stackrel{\wedge}{\sigma}^{(q)}_{Y,n}
 -q!\stackrel{\wedge}{\sigma}^{(1)}_{Y,n}\wedge
 \ldots \wedge\stackrel{\wedge}{\sigma}^{(1)}_{Y,n}\right)
 \left(\Id^{\otimes(q-1)}
 \otimes (\Id-s^{\wedge}_{Y,n+1}\rho_{Y})\right) A^{(q)}_{\cH_{Y}}}
 \end{align*}
 \begin{align*}
 & \leq \Tr\abs{\left(\stackrel{\wedge}{\sigma}^{(q)}_{Y,n}
 -q!\stackrel{\wedge}{\sigma}^{(1)}_{Y,n}\wedge
 \ldots \wedge\stackrel{\wedge}{\sigma}^{(1)}_{Y,n}\right)
 \left(\Id^{\otimes(q-1)}
 \otimes (\Id+s^{\wedge}_{Y,n+1}\rho_{Y})\right)
 A^{(q)}_{\cH_{Y}}} \\
 & \quad +\norm{\Id-s^{\wedge}_{Y,n+1}\rho_{Y}}
 \,\Tr\abs{\stackrel{\wedge}{\sigma}^{(q)}_{Y,n}
 -q!\stackrel{\wedge}{\sigma}^{(1)}_{Y,n}\wedge
 \ldots \wedge\stackrel{\wedge}{\sigma}^{(1)}_{Y,n}},
 \end{align*}
hence
 \begin{multline}{\label{glowne12}}
 \left(2-\norm{\Id-s^{\wedge}_{Y,n+1}\rho_{Y}}\right)
 \,\Tr\abs{\stackrel{\wedge}{\sigma}^{(q)}_{Y,n}
 -q!\stackrel{\wedge}{\sigma}^{(1)}_{Y,n}\wedge
 \ldots \wedge\stackrel{\wedge}{\sigma}^{(1)}_{Y,n}} \\
 \leq \Tr\abs{\left(\stackrel{\wedge}{\sigma}^{(q)}_{Y,n}
 -q!\stackrel{\wedge}{\sigma}^{(1)}_{Y,n}\wedge
 \ldots \wedge\stackrel{\wedge}{\sigma}^{(1)}_{Y,n}\right)
 \left(\Id^{\otimes(q-1)}
 \otimes (\Id+s^{\wedge}_{Y,n+1}\rho_{Y})\right)
 A^{(q)}_{\cH_{Y}}}.
 \end{multline}
Since the operators $\rho_{Y}$ are trace class,
 $\displaystyle\inf_{\varphi\in\cH_{Y};\,\norm{\varphi}=1}
 \inn{\varphi}{\rho_{Y}\varphi}=0.$
Thus, by assumption~\eqref{assferm} and
the selfadjointness of the operators
 $\Id-s^{\wedge}_{Y,n+1}\rho_{Y},$
one obtains
 \begin{align}{\label{glowne13}}
 &\nonumber \norm{\Id-s^{\wedge}_{Y,n+1}\rho_{Y}}
 =\sup_{\substack{\varphi\in\cH_{Y} \\ \norm{\varphi}=1}}
 \abs{\inn{\varphi}{(\Id-s^{\wedge}_{Y,n+1}\rho_{Y})\varphi}} \\
 & =\max\set{1-s^{\wedge}_{Y,n+1}
 \inf_{\substack{\varphi\in\cH_{Y} \\ \norm{\varphi}=1}}
 \inn{\varphi}{\rho_{Y}\varphi}, \; s^{\wedge}_{Y,n+1}
 \sup_{\substack{\varphi\in\cH_{Y} \\ \norm{\varphi}=1}}
 \inn{\varphi}{\rho_{Y}\varphi}-1}=1.
 \end{align}

The rest of the proof of~\eqref{glowne2} is
by induction with respect to $k\geq 2.$
\newline
$1^{\circ}.$ ($k=2$) By Theorem~\ref{rek} for
$n\geq 2$ one has
 \begin{align}{\label{glowne5}}
 \nonumber \frac{1}{(n+1)^{2}}\newt{n+1}{2}
 \stackrel{\wedge}{\sigma}^{(2)}_{Y,n+1}
 & =\frac{n}{n+1}\stackrel{\wedge}{\sigma}^{(1)}_{Y,n}
 \wedge \left((n+1)^{-1}s^{\wedge}_{Y,n+1}\rho_{Y}\right) \\
 &\quad -\frac{1}{(n+1)^{2}}\newt{n}{2}
 \stackrel{\wedge}{\sigma}^{(2)}_{Y,n}
 \left(\Id\otimes (s^{\wedge}_{Y,n+1}\rho_{Y})\right)
 A^{(2)}_{\cH_{Y}}.
 \end{align}
Assumption~\eqref{assferm} gives
 $\Tr\abs{\stackrel{\wedge}{\sigma}^{(2)}_{Y,n}
 \left(\Id\otimes (s^{\wedge}_{Y,n+1}\rho_{Y})\right)
 A^{(2)}_{\cH_{Y}}} \leq s^{\wedge}_{Y,n+1}\norm{\rho_{Y}}
 \,\Tr\abs{\stackrel{\wedge}{\sigma}^{(2)}_{Y,n}} \leq 2,$
hence, by Eq.~\eqref{glowne5}, Remark~\ref{strongprop}, and
the assumption
 $\stackrel{\wedge}{\sigma}^{(2)}_{Y,n}
 \approx \stackrel{\wedge}{\sigma}^{(2)}_{Y,n+1},$
one obtains
 \begin{equation*}
 \stackrel{\wedge}{\sigma}^{(2)}_{Y,n}
 +\stackrel{\wedge}{\sigma}^{(2)}_{Y,n}
 \left(\Id\otimes (s^{\wedge}_{Y,n+1}\rho_{Y})\right)
 A^{(2)}_{\cH_{Y}}
 \approx 2\frac{n}{n+1}\stackrel{\wedge}{\sigma}^{(1)}_{Y,n}
 \wedge\left((n+1)^{-1}s^{\wedge}_{Y,n+1}\rho_{Y}\right).
 \end{equation*}
Thus, in view of equivalence~\eqref{as1} from Theorem~\ref{as}
and Lemma~\ref{iloczyny}, one has
 \begin{multline}{\label{rown1}}
 \stackrel{\wedge}{\sigma}^{(2)}_{Y,n}
 +\stackrel{\wedge}{\sigma}^{(2)}_{Y,n}
 \left(\Id\otimes (s^{\wedge}_{Y,n+1}\rho_{Y})\right)
 A^{(2)}_{\cH_{Y}} \\
 \approx 2\frac{n}{n+1}\stackrel{\wedge}{\sigma}^{(1)}_{Y,n}
 \wedge\left(\stackrel{\wedge}{\sigma}^{(1)}_{Y,n}
 (\Id+s^{\wedge}_{Y,n+1}\rho_{Y})\right).
 \end{multline}
Furthermore, assumption~\eqref{assferm} implies that
the trace norms of the operators on the r.h.s of~\eqref{rown1}
are uniformly bounded.
Therefore, according to Remark~\ref{strongprop},
 \begin{equation}{\label{glowne6}}
 \left(\stackrel{\wedge}{\sigma}^{(2)}_{Y,n}
 -2\stackrel{\wedge}{\sigma}^{(1)}_{Y,n}\wedge
 \stackrel{\wedge}{\sigma}^{(1)}_{Y,n}\right)
 \left(\Id\otimes (\Id+s^{\wedge}_{Y,n+1}\rho_{Y})\right)
 A^{(2)}_{\cH_{Y}}\approx 0.
 \end{equation}
The explicit form of
 $\stackrel{\wedge}{\sigma}^{(2)}_{Y,n},$
 $\stackrel{\wedge}{\sigma}^{(1)}_{Y,n}
 \wedge \stackrel{\wedge}{\sigma}^{(1)}_{Y,n}$
given by Theorem~\ref{jawny} implies that
 $\stackrel{\wedge}{\sigma}^{(2)}_{Y,n}
 -2\stackrel{\wedge}{\sigma}^{(1)}_{Y,n}
 \wedge \stackrel{\wedge}{\sigma}^{(1)}_{Y,n}$
and
 $\left(\Id\otimes (\Id+s^{\wedge}_{Y,n+1}\rho_{Y})\right)
 A^{(2)}_{\cH_{Y}}$
commute, which proves
the selfadjointness of the operator on the l.h.s of~\eqref{glowne6}.
Thus conditions~\eqref{glowne6}, \eqref{glowne12} for $q=2$, \eqref{glowne13},
and Lemma~\ref{mocna} yield relation~\eqref{glowne2} for $k=2.$
\newline
\noindent $2^{\circ}.$ Assuming validity of
equivalence~\eqref{glowne2} for $k\in\set{2,\ldots,q-1},$ where
$q\in\N,$ $q>2,$ its validity will be proved for $k=q.$

By Theorem~\ref{rek} for $n\geq q$ one has
 \begin{align}{\label{glowne9}}
 \nonumber \frac{1}{(n+1)^{q}}\newt{n+1}{q}
 \stackrel{\wedge}{\sigma}^{(q)}_{Y,n+1}
 & =\frac{1}{(n+1)^{q-1}}\newt{n}{q-1}
 \stackrel{\wedge}{\sigma}^{(q-1)}_{Y,n}
 \wedge \left((n+1)^{-1}s^{\wedge}_{Y,n+1}\rho_{Y}\right) \\
 & \quad -\frac{1}{(n+1)^{q}}\newt{n}{q}
 \stackrel{\wedge}{\sigma}^{(q)}_{Y,n}
 \left(\Id^{\otimes(q-1)}
 \otimes (s^{\wedge}_{Y,n+1}\rho_{Y})\right) A^{(q)}_{\cH_{Y}}.
 \end{align}
Assumption~\eqref{assferm} implies
 \begin{equation*}
 \Tr\abs{\stackrel{\wedge}{\sigma}^{(q)}_{Y,n}
 \left(\Id^{\otimes(q-1)}
 \otimes (s^{\wedge}_{Y,n+1}\rho_{Y})\right)
 A^{(q)}_{\cH_{Y}}} \leq s^{\wedge}_{Y,n+1}\norm{\rho_{Y}}
 \,\Tr\abs{\stackrel{\wedge}{\sigma}^{(q)}_{Y,n}} \leq 2,
 \end{equation*}
hence, in view of Eq.~\eqref{glowne9}, Remark~\ref{strongprop},
and the assumption
 $\stackrel{\wedge}{\sigma}^{(q)}_{Y,n}
 \approx \stackrel{\wedge}{\sigma}^{(q)}_{Y,n+1},$
 \begin{multline*}
 \stackrel{\wedge}{\sigma}^{(q)}_{Y,n}
 +\stackrel{\wedge}{\sigma}^{(q)}_{Y,n}
 \left(\Id^{\otimes(q-1)}
 \otimes (s^{\wedge}_{Y,n+1}\rho_{Y})\right) A^{(q)}_{\cH_{Y}} \\
 \approx \frac{q!}{(n+1)^{q-1}}\newt{n}{q-1}
 \stackrel{\wedge}{\sigma}^{(q-1)}_{Y,n}
 \wedge\left((n+1)^{-1}s^{\wedge}_{Y,n+1}\rho_{Y}\right).
 \end{multline*}
Thus, by relation~\eqref{as1} from Theorem~\ref{as},
Lemma~\ref{iloczyny}, and Remark~\ref{strongprop}, one has
 \begin{multline}{\label{rown2}}
 \stackrel{\wedge}{\sigma}^{(q)}_{Y,n}
 +\stackrel{\wedge}{\sigma}^{(q)}_{Y,n}
 \left(\Id^{\otimes(q-1)}
 \otimes (s^{\wedge}_{Y,n+1}\rho_{Y})\right)  A^{(q)}_{\cH_{Y}} \\
 \approx \frac{q!}{(q-1)!}\stackrel{\wedge}{\sigma}^{(q-1)}_{Y,n}
 \wedge\left(\stackrel{\wedge}{\sigma}^{(1)}_{Y,n}
 (\Id+s^{\wedge}_{Y,n+1}\rho_{Y})\right),
 \end{multline}
since the trace norms
of the operators on the r.h.s. of~\eqref{rown2} are uniformly
bounded, by assumption~\eqref{assferm}.
Furthermore, in view of Lemma~\ref{iloczyny} and the inductive
hypothesis
 $\displaystyle \stackrel{\wedge}{\sigma}^{(q-1)}_{Y,n}
 \approx (q-1)!\underbrace{\stackrel{\wedge}{\sigma}^{(1)}_{Y,n}
 \wedge \ldots \wedge\stackrel{\wedge}{\sigma}^{(1)}_{Y,n}}_{q-1},$
condition~\eqref{rown2} yields
 \begin{multline*}
 \stackrel{\wedge}{\sigma}^{(q)}_{Y,n}
 +\stackrel{\wedge}{\sigma}^{(q)}_{Y,n}
 \left(\Id^{\otimes(q-1)}
 \otimes (s^{\wedge}_{Y,n+1}\rho_{Y})\right) A^{(q)}_{\cH_{Y}} \\
 \approx q!\underbrace{\stackrel{\wedge}{\sigma}^{(1)}_{Y,n}\wedge
 \ldots \wedge\stackrel{\wedge}{\sigma}^{(1)}_{Y,n}}_{q-1}
 \wedge\left(\stackrel{\wedge}{\sigma}^{(1)}_{Y,n}
 (\Id+s^{\wedge}_{Y,n+1}\rho_{Y})\right),
 \end{multline*}
hence
 \begin{equation}{\label{glowne10}}
 \left(\stackrel{\wedge}{\sigma}^{(q)}_{Y,n}
 -q!\underbrace{\stackrel{\wedge}{\sigma}^{(1)}_{Y,n}\wedge
 \ldots \wedge\stackrel{\wedge}{\sigma}^{(1)}_{Y,n}}_{q}\right)
 \left(\Id^{\otimes(q-1)}
 \otimes (\Id+s^{\wedge}_{Y,n+1}\rho_{Y})\right)
 A^{(q)}_{\cH_{Y}}\approx 0.
 \end{equation}
From the explicit form of
 $\stackrel{\wedge}{\sigma}^{(q)}_{Y,n},$
 $\stackrel{\wedge}{\sigma}^{(1)}_{Y,n}\wedge
 \ldots \wedge\stackrel{\wedge}{\sigma}^{(1)}_{Y,n}$
given by Theorem~\ref{jawny} one finds that
 $\stackrel{\wedge}{\sigma}^{(q)}_{Y,n}
 -q!\stackrel{\wedge}{\sigma}^{(1)}_{Y,n}\wedge
 \ldots \wedge\stackrel{\wedge}{\sigma}^{(1)}_{Y,n}$
and
 $\left(\Id^{\otimes(q-1)}
 \otimes (\Id+s^{\wedge}_{Y,n+1}\rho_{Y})\right)
 A^{(q)}_{\cH_{Y}}$
commute, which proves the selfadjointness of the operator
on the l.h.s of~\eqref{glowne10}.
Thus conditions~\eqref{glowne10}, \eqref{glowne12}, \eqref{glowne13},
and Lemma~\ref{mocna} yield
 $\stackrel{\wedge}{\sigma}^{(q)}_{Y,n}
 \approx q!\underbrace{\stackrel{\wedge}{\sigma}^{(1)}_{Y,n}\wedge
 \ldots \wedge\stackrel{\wedge}{\sigma}^{(1)}_{Y,n}}_{q}.$
Validity of relation~\eqref{glowne2} has been proved.

Now turn to equivalence~\eqref{glowne4}. Similarly
to~\eqref{glowne12} one has
\begin{multline*}
 \left(2-\norm{\Id+s^{\vee}_{Y,n+1}\rho_{Y}}\right)
 \,\Tr\abs{\stackrel{\vee}{\sigma}^{(q)}_{Y,n}
 -q!\stackrel{\vee}{\sigma}^{(1)}_{Y,n}\vee
 \cdots \vee\stackrel{\vee}{\sigma}^{(1)}_{Y,n}} \\
 \leq \Tr\abs{\left(\stackrel{\vee}{\sigma}^{(q)}_{Y,n}
 -q!\stackrel{\vee}{\sigma}^{(1)}_{Y,n}\vee
 \cdots \vee\stackrel{\vee}{\sigma}^{(1)}_{Y,n}\right)
 \left(\Id^{\otimes(q-1)}
 \otimes (\Id-s^{\vee}_{Y,n+1}\rho_{Y})\right) S^{(q)}_{\cH_{Y}}}.
 \end{multline*}
Furthermore, according to assumption~\eqref{assbos},
 \begin{equation}{\label{epsestpr}}
 2-\norm{\Id+s^{\vee}_{Y,n+1}\rho_{Y}}
 \geq 2-(1+s^{\vee}_{Y,n+1}\norm{\rho_{Y}}) \geq 1-\epsilon>0.
 \end{equation}
The rest of the proof of~\eqref{glowne4} is by induction
with respect to $k\geq 2$ and proceeds analogously
to the proof of~\eqref{glowne2} with condition~\eqref{glowne13}
replaced by~\eqref{epsestpr} and the operators
 $\Id\mp s^{\wedge}_{Y,n+1}\rho_{Y}$
replaced by
 $\Id\pm s^{\vee}_{Y,n+1}\rho_{Y}$
(inversion of signs).
\end{proof}

Theorem~\ref{glowne} allows to replace \mbox{$(n,k)$-con}tractions
of antisymmetric and symmetric pro\-duct density operators by
antisymmetric and symmetric products of \mbox{$1$-par}ticle
contractions, respectively, if the number $n$ of particles in the
system is large. Further simplification, consisting in replacement
of antisymmetric and symmetric products by tensor products, will be
now proved possible. To this end weaker conditions on the
asymptotic equivalence relation will be imposed.

 \begin{definition}{\label{rel}}
 Fix $k\in\N$ and $d\in\R,$ $d>0.$ Families $\rodz{A}{Y}{n},$
 $\rodz{B}{Y}{n}$ of operators
  $A_{Y,n},B_{Y,n}\in\cT\left(\cH_{Y}^{\otimes k}\right)$
 are called \emph{weakly asymptotically equivalent}
 (symbolically: $A_{Y,n}\sim B_{Y,n}$), if
  $\displaystyle\limt{n}{Y}\Tr\,(A_{Y,n}-B_{Y,n})C_{Y,n}=0$
 for every family $\rodz{C}{Y}{n}$ of operators of the form
  $\displaystyle C_{Y,n} =\bigotimes_{i=1}^{k}C_{Y,n}^{(i)},$
 where
  $C_{Y,n}^{(i)}\in\cB\left(\cH_{Y}\right)$
 ($i\in\set{1,\ldots,k},$ $(Y,n)\in\cM(\Omega)\times\N$) are
 operators with uniformly bounded operator norms.
 \end{definition}

The relation $\sim$ has the properties analogous to
the properties of
the relation $\approx$ from Remark~\ref{strongprop}.

 \begin{definition}{\label{cykliczny}}
 Let $k\in\N,$ $k\geq 2.$ Fix $\pi\in S_k.$ A~set
 $X\subset\set{1,\ldots,k}$ is called a~\emph{cyclic set of the
 permutation $\pi$}, if $X=\set{l_1,\ldots,l_q}$ for some
  $l_1,\ldots, l_q\in\set{1,\ldots,k},$ $q\in\set{2,\ldots,k},$
 such that $\pi(l_{s})=l_{s+1}$ for every $s\in\set{1,\ldots,q-1},$
 and $\pi(l_{q})=l_{1}.$ A~singleton
 $\set{l}\subset\set{1,\ldots,k}$ such that $\pi(l)=l$ is also
 called a~cyclic set of the permutation $\pi.$
 \end{definition}

Note that the set $\set{1,\ldots,k}$ from the above definition can
be represented as the union of disjoint cyclic sets of $\pi.$

 \begin{lemma}{\label{pars}}
 Let $k\in\N,$ $k\geq 2.$ If
  $B^{(1)},\ldots,B^{(k)}\in\cT\left(\cH_{Y}\right)$
 then
 \begin{equation}{\label{pars1}}
 k!\,\Tr\left(B_{Y,n}^{(1)}\otimes
 \cdots\otimes B_{Y,n}^{(k)}\right) A^{(k)}_{\cH_{Y}}
 =\sum_{\pi\in S_{k}}\sgn\,\pi\prod_{j=1}^{p(\pi)}
 \Tr\prod_{s=1}^{q_{j}}B_{Y,n}^{(l_{j,s})},
 \end{equation}
 where $p(\pi)\in\set{1,\ldots,k}$ is the number of disjoint cyclic
 sets of $\pi,$ indexed by $j,$ and $q_{j}$ denotes the number of
 elements of the \mbox{$j$th} cyclic set of $\pi,$ which is
 $\set{l_{j,1},\ldots,l_{j,q_{j}}},$
 where
 \begin{equation}{\label{lemcyclord}}
 \pi(l_{j,q_j})=l_{j,1} \quad \text{and, for $q_j\geq 2,$} \quad
 \pi(l_{j,s})=l_{j,s+1}, \quad s=1,\ldots,q_j-1.
 \end{equation}
 Clearly,
 $\displaystyle \sum_{j=1}^{p(\pi)}q_{j}=k$
 and
 $\displaystyle
 \bigcup_{j=1}^{p(\pi)}\bigcup_{s=1}^{q_{j}}\set{l_{j,s}}
 =\set{1,\ldots,k}.$
 \end{lemma}
\begin{proof} Let $\set{\varphi_{i}}_{i\in\N}$
be an orthonormal basis of $\cH_{Y}.$ One has
 \begin{equation}{\label{pars6}}
 k!\,\Tr\left(B_{Y,n}^{(1)}\otimes
 \cdots\otimes B_{Y,n}^{(k)}\right) A^{(k)}_{\cH_{Y}}
 =\sum_{\pi\in S_{k}}\sgn\,\pi\prod_{j=1}^{p(\pi)}M_j,
 \end{equation}
where
 \begin{equation*}
 M_j\eqdf\sum_{i_{l_{j,1}}=1}^{\infty}
 \ldots\sum_{i_{l_{j,q_{j}}}=1}^{\infty}
 \inn{\varphi_{i_{l_{j,1}}}}{B^{(l_{j,1})}
 \varphi_{i_{\pi(l_{j,1})}}}
 \cdots \inn{\varphi_{i_{l_{j,q_{j}}}}}{B^{(l_{j,q_{j}})}
 \varphi_{i_{\pi(l_{j,q_{j}})}}}.
 \end{equation*}
If $q_{j}>2$ for some $j\in\set{1,\ldots,p(\pi)}$ then, by
 condition~\eqref{lemcyclord} and Parseval's formula,
 \begin{align*}
 M_j  &=\sum_{i_{l_{j,1}}=1}^{\infty}
 \ldots \sum_{i_{l_{j,q_{j}}}=1}^{\infty} \\
 &\quad \inn{\varphi_{i_{l_{j,1}}}}{B^{(l_{j,1})}
 \varphi_{i_{l_{j,2}}}}\inn{\varphi_{i_{l_{j,2}}}}{B^{(l_{j,2})}
 \varphi_{i_{l_{j,3}}}}
 \cdots \cdots \inn{\varphi_{i_{l_{j,q_{j}}}}}{B^{(l_{j,q_{j}})}
 \varphi_{i_{l_{j,1}}}}
 \end{align*}
 \begin{align*}
 &=\sum_{i_{l_{j,1}}=1}^{\infty} \sum_{i_{l_{j,3}}=1}^{\infty}
 \ldots \sum_{i_{l_{j,q_{j}}}=1}^{\infty} \\
 &\qquad \inn{\varphi_{i_{l_{j,1}}}}{B^{(l_{j,1})}B^{(l_{j,2})}
 \varphi_{i_{l_{j,3}}}}
 \cdots \inn{\varphi_{i_{l_{j,q_{j}}}}}{B^{(l_{j,q_{j}})}
 \varphi_{i_{l_{j,1}}}}.
 \end{align*}
Performing successive summations one then obtains
 \begin{equation*}
 M_j =\sum_{i_{l_{j,1}}=1}^{\infty}
 \inn{\varphi_{i_{l_{j,1}}}}{\left(
 \prod_{s=1}^{q_{j}}B^{(l_{j,s})}\right)
 \varphi_{i_{l_{j,1}}}} =\Tr\prod_{s=1}^{q_{j}}B^{(l_{j,s})}.
 \end{equation*}
The derivation of the above formula for $q_j=1,2,$ after
simplifications, proceeds analogously. This completes the proof
of Eq.~\eqref{pars1}, in view of Eq.~\eqref{pars6}.
\end{proof}

 \begin{lemma}{\label{asnorm}}
 One has
 \begin{equation}{\label{normasferm}}
 \limt{n}{Y}\norm{\stackrel{\wedge}{\sigma}^{(1)}_{Y,n}}=0,
 \end{equation}
 and if
  $\stackrel{\vee}{\sigma}^{(2)}_{Y,n}
  \approx 2\stackrel{\vee}{\sigma}^{(1)}_{Y,n}
  \vee\stackrel{\vee}{\sigma}^{(1)}_{Y,n}$
 (see Theorem~\ref{glowne}) then
 \begin{equation}{\label{normasbos}}
 \limt{n}{Y}\norm{\stackrel{\vee}{\sigma}^{(1)}_{Y,n}}=0.
 \end{equation}
 \end{lemma}
\begin{proof}
To prove Eq.~\eqref{normasferm} it suffices to observe that,
according to Theorem~\ref{colemferm},
 \begin{equation*}
 \norm{\stackrel{\wedge}{\sigma}^{(1)}_{Y,n}}
 =\norm{\Lc^{1}_{n} \left(\frac{1}{\xi^{\wedge}_{Y,n}}\rho_{Y}^{\wedge n}\right)}
 \leq\frac{1}{n}\frac{1}{\xi^{\wedge}_{Y,n}}
 \norm{\rho_{Y}^{\wedge n}}
 \leq\frac{1}{n}\frac{1}{\xi^{\wedge}_{Y,n}}\Tr\rho_{Y}^{\wedge n}
 =\frac{1}{n}.
 \end{equation*}

Now Eq.~\eqref{normasbos} will be proved. Let
$\set{\varphi_{i}}_{i\in\N}$ be an orthonormal basis of $\cH_{Y}$
for fixed $Y\in\cM(\Omega).$ Then
 \begin{align}{\label{asnorm6}}
 \nonumber \Tr\,2\stackrel{\vee}{\sigma}^{(1)}_{Y,n}
 \vee\stackrel{\vee}{\sigma}^{(1)}_{Y,n}
 &=2\,\Tr\left(\stackrel{\vee}{\sigma}^{(1)}_{Y,n}
 \otimes\stackrel{\vee}{\sigma}^{(1)}_{Y,n}\right)
 S^{(2)}_{\cH_{Y}} \\
 &\nonumber =\sum_{\pi\in S_{2}}\sum_{i_{1},i_{2}=1}^{\infty}
 \inn{\varphi_{i_{1}}}{\stackrel{\vee}{\sigma}^{(1)}_{Y,n}
 \varphi_{i_{\pi(1)}}}
 \inn{\varphi_{i_{2}}}{\stackrel{\vee}{\sigma}^{(1)}_{Y,n}
 \varphi_{i_{\pi(2)}}} \\
 & =\left(\Tr\,\stackrel{\vee}{\sigma}^{(1)}_{Y,n}\right)^{2}
 +\Tr\left(\stackrel{\vee}{\sigma}^{(1)}_{Y,n}
 \stackrel{\vee}{\sigma}^{(1)}_{Y,n}\right).
 \end{align}
Taking into account Eq.~\eqref{asnorm6}, the relation
 $\stackrel{\vee}{\sigma}^{(2)}_{Y,n}
 \approx 2\stackrel{\vee}{\sigma}^{(1)}_{Y,n}
 \vee\stackrel{\vee}{\sigma}^{(1)}_{Y,n},$
Definition~\ref{relacja} for
 $C_{Y,n} =\Id^{\otimes 2},$
and the equality
 $\Tr\,\stackrel{\vee}{\sigma}^{(2)}_{Y,n}
 =\Tr\,\stackrel{\vee}{\sigma}^{(1)}_{Y,n}=1,$
one obtains
 \begin{equation}{\label{asnorm7}}
 \limt{n}{Y}\Tr\,\left(\stackrel{\vee}{\sigma}^{(1)}_{Y,n}
 \stackrel{\vee}{\sigma}^{(1)}_{Y,n}\right)=0.
 \end{equation}
Furthermore,
 \begin{equation*}
 \norm{\stackrel{\vee}{\sigma}^{(1)}_{Y,n}\varphi}^{2}
 =\inn{\varphi}{\stackrel{\vee}{\sigma}^{(1)}_{Y,n}
 \stackrel{\vee}{\sigma}^{(1)}_{Y,n}\varphi}
 \leq\Tr\,\left(\stackrel{\vee}{\sigma}^{(1)}_{Y,n}
 \stackrel{\vee}{\sigma}^{(1)}_{Y,n}\right)
 \end{equation*}
for every $\varphi\in\cH_{Y}$ such that $\norm{\varphi}=1,$
hence Eq.~\eqref{asnorm7} yields Eq.~\eqref{normasbos}.
\end{proof}

Notice that Eq.~\eqref{normasferm} can be also proved analogously
to Eq.~\eqref{normasbos} under the additional assumption
 $\stackrel{\wedge}{\sigma}^{(2)}_{Y,n}
 \approx 2\stackrel{\wedge}{\sigma}^{(1)}_{Y,n}
 \wedge\stackrel{\wedge}{\sigma}^{(1)}_{Y,n}.$

The proof of the next theorem for $k=2$ was given
in~\cite{KossakowskiRMP86,MackowiakPR99}.

 \begin{theorem}[Asymptotic formulae II]{\label{zmiana}}
 Let $k\in\N,$ $k\geq 2.$ One has
 \begin{equation}{\label{zmiana1}}
 k!\underbrace{\stackrel{\wedge}{\sigma}^{(1)}_{Y,n}\wedge
 \ldots \wedge\stackrel{\wedge}{\sigma}^{(1)}_{Y,n}}_{k}
 \sim \underbrace{\stackrel{\wedge}{\sigma}^{(1)}_{Y,n}\otimes
 \cdots \otimes\stackrel{\wedge}{\sigma}^{(1)}_{Y,n}}_{k},
 \end{equation}
 and if
  $\displaystyle \limt{n}{Y}
  \norm{\stackrel{\vee}{\sigma}^{(1)}_{Y,n}}=0$
 (see Lemma~\ref{asnorm}) then
 \begin{equation}{\label{zmiana2}}
 k!\underbrace{\stackrel{\vee}{\sigma}^{(1)}_{Y,n}\vee
 \cdots \vee\stackrel{\vee}{\sigma}^{(1)}_{Y,n}}_{k}
 \sim \underbrace{\stackrel{\vee}{\sigma}^{(1)}_{Y,n}\otimes
 \cdots \otimes\stackrel{\vee}{\sigma}^{(1)}_{Y,n}}_{k}.
 \end{equation}
 \end{theorem}
\begin{proof}
First Eq.~\eqref{zmiana1} will be proved. Fix a~family
 $\rodz{C}{Y}{n}$
of operators such as in Definition~\ref{rel} and set
 \begin{equation*}
 B_{Y,n}^{(r)}
 \eqdf\stackrel{\wedge}{\sigma}^{(1)}_{Y,n}C_{Y,n}^{(r)},
 \quad r=1,\ldots,k.
 \end{equation*}
Then, by Lemma~\ref{pars}, one has
 \begin{align*}
 & \Tr\,k!\left(\stackrel{\wedge}{\sigma}^{(1)}_{Y,n}\wedge
 \ldots \wedge\stackrel{\wedge}{\sigma}^{(1)}_{Y,n}\right)
 \left(C_{Y,n}^{(1)}\otimes\cdots\otimes C_{Y,n}^{(k)}\right) \\
 & =k!\,\Tr\left(B_{Y,n}^{(1)}\otimes
 \cdots\otimes B_{Y,n}^{(k)}\right) A^{(k)}_{\cH_{Y}}
 =\sum_{\pi\in S_{k}}\sgn\,\pi\prod_{j=1}^{p(\pi)}
 \Tr\prod_{s=1}^{q_{j}}B_{Y,n}^{(l_{j,s})}
 \end{align*}
 \begin{align*}
 & =\Tr\left(\stackrel{\wedge}{\sigma}^{(1)}_{Y,n}\otimes
 \cdots \otimes\stackrel{\wedge}{\sigma}^{(1)}_{Y,n}\right)
 \left(C_{Y,n}^{(1)}\otimes\cdots\otimes C_{Y,n}^{(k)}\right) \\
 &\quad +\sum_{\substack{\pi\in S_{k} \\ \pi\not =\mathrm{Id}}}\sgn\,\pi
 \prod_{j=1}^{p(\pi)}
 \Tr\prod_{s=1}^{q_{j}}B_{Y,n}^{(l_{j,s})}.
 \end{align*}
Thus
 \begin{multline}{\label{zmiana3}}
 \Tr\left(k!\stackrel{\wedge}{\sigma}^{(1)}_{Y,n}\wedge
 \ldots \wedge\stackrel{\wedge}{\sigma}^{(1)}_{Y,n}
 -\stackrel{\wedge}{\sigma}^{(1)}_{Y,n}\otimes
 \cdots\otimes\stackrel{\wedge}{\sigma}^{(1)}_{Y,n}\right)
 \left(C_{Y,n}^{(1)}\otimes\cdots\otimes C_{Y,n}^{(k)}\right) \\
 =\sum_{\substack{\pi\in S_{k} \\ \pi\not=\mathrm{Id}}}\sgn\,\pi
 \prod_{j=1}^{p(\pi)}
 \Tr\prod_{s=1}^{q_{j}}B_{Y,n}^{(l_{j,s})}.
 \end{multline}
Now, let $\pi\in S_{k},$ $\pi\not=\mathrm{Id},$ be fixed. If
$q_{j}=1$ for some $j\in\set{1,\ldots,p(\pi)}$ then
 \begin{equation*}
 \abs{\Tr\prod_{s=1}^{q_{j}}B_{Y,n}^{(l_{j,s})}}
 \equiv\abs{\Tr\,B_{Y,n}^{(l_{j,1})}}
 \leq\norm{C_{Y,n}^{(l_{j,1})}}
 \,\Tr\abs{\stackrel{\wedge}{\sigma}^{(1)}_{Y,n}}
 =\norm{C_{Y,n}^{(l_{j,1})}},
 \end{equation*}
whereas if $q_{j}\geq 2$ then
 \begin{align*}
 \abs{\Tr\prod_{s=1}^{q_{j}}B_{Y,n}^{(l_{j,s})}}
 &\leq\norm{\prod_{s=1}^{q_{j}-1}B_{Y,n}^{(l_{j,s})}}
 \,\Tr\abs{B_{Y,n}^{(l_{j,q_{j}})}} \\
 & \leq\norm{\stackrel{\wedge}{\sigma}^{(1)}_{Y,n}}^{q_{j}-1}
 \left(\prod_{s=1}^{q_{j}-1}\norm{C_{Y,n}^{(l_{j,s})}}\right)
 \norm{C_{Y,n}^{(l_{j,q_{j}})}}
 \,\Tr\abs{\stackrel{\wedge}{\sigma}^{(1)}_{Y,n}} \\
 & \leq\norm{\stackrel{\wedge}{\sigma}^{(1)}_{Y,n}}^{q_{j}-1}
 \prod_{s=1}^{q_{j}}\norm{C_{Y,n}^{(l_{j,s})}}.
 \end{align*}
Since $\pi\not=\mathrm{Id},$ there exists at least one
 $j\in\set{1,\ldots,p(\pi)}$
such that $q_{j}\geq 2,$ hence
 \begin{align*}
 \abs{\prod_{j=1}^{p(\pi)}
 \Tr\prod_{s=1}^{q_{j}}B_{Y,n}^{(l_{j,s})}}
 & \leq\left(\prod_{j=1}^{p(\pi)}
 \norm{\stackrel{\wedge}{\sigma}^{(1)}_{Y,n}}^{q_{j}-1}\right)
 \prod_{j=1}^{p(\pi)} \prod_{s=1}^{q_{j}}
 \norm{C_{Y,n}^{(l_{j,s})}} \\
 & =\norm{C_{Y,n}}\prod_{j=1}^{p(\pi)}
 \norm{\stackrel{\wedge}{\sigma}^{(1)}_{Y,n}}^{q_{j}-1}
 \end{align*}
and at least one exponent $q_{j}-1$ is nonzero. Thus, by the
uniform boundedness of the norms $\norm{C_{Y,n}}$ and
Lemma~\ref{asnorm}, the termodynamic limit of the l.h.s
of~\eqref{zmiana3} equals $0,$
which proves the validity of relation~\eqref{zmiana1}.

The proof of relation~\eqref{zmiana2}, after discarding the
permutation signs and replacing $\wedge$ by $\vee,$ proceeds
analogously.
\end{proof}

\appendix

\section{Product integral kernels of trace class
operators}{\label{sectjad}}

In this section theorems concerning product integral kernels,
exploited in Section~\ref{potkontr}, are formulated.

Fix the Hilbert space $\cH_Y \eqdf L^{2}(Y,\mu)$ over the field
$\K=\C$ or $\R,$ where the measure $\mu$ is separable and
\mbox{$\sigma$-fi}nite. For every $n\in\N$ the space
 $\cH_Y^{\otimes n}$ is identified with
 $L^2(Y^n,\mu^{\otimes n}).$
Unless otherwise stated,
elements of $L^2$ spaces are identified with their
representatives and denoted by the same symbols.

Let $\cK\in L^2(Y^2,\mu^{\otimes 2}).$ In the case of the
integral operator
 $K\colon H_Y\to H_Y$
defined for every $\varphi\in H_Y$ and
\mbox{$\mu$-a.a.} $x\in Y$ by
 \begin{equation}{\label{intdef}}
 (K\varphi)(x) =\int_{Y}\cK(x,y)\varphi(y)\,\dx\mu(y)
 \end{equation}
both $\cK$ regarded as an element of $L^2(Y^2,\mu^{\otimes 2})$
as well as its arbitrary representative is called
an \emph{integral kernel of $K$}. The kernel $\cK$ is unique
as an element of $L^2(Y^2,\mu^{\otimes 2})$ but a~representative of
$\cK$ of a~special form,
given in Lemma~\ref{tr} and Definition~\ref{ilkern}, is useful in
computations of the trace of $K.$

Let $\cHS(\cH_Y)$ be the space of Hilbert-Schmidt operators on
$\cH_Y$ with the inner product defined by
 $\inn{A}{B}_{\cHS(\cH_Y)} \eqdf\Tr\,A^{\ast}B$
and the induced norm denoted by $\norm{\cdot}_{\cHS(\cH_Y)}.$
In the sequel use is made of the following theorem, the proof of
which can be found in~\cite{SchattenSV60}.

 \begin{theorem}{\label{hscalk}}
 An operator $K\in\cB(H_Y)$ is
 Hilbert-Schmidt iff it is an integral operator with an
 integral kernel
  $\cK\in L^2(Y^2,\mu^{\otimes 2}).$
 Furthermore,
  $\norm{K}_{\cHS(H_Y)}
  =\norm{\cK}_{L^2(Y^2,\mu^{\otimes 2})}.$
 \end{theorem}

 \begin{corollary}{\label{hscalkcor}}
 Let $K,G\in\cHS(H_Y)$ and let
  $\cK,\cG\in L^2(Y^2,\mu^{\otimes 2})$
 be integral kernels of the operators $K,$ $G,$ respectively. Then
  $\inn{K}{G}_{\cHS(H_Y)}
  =\inn{\cK}{\cG}_{L^2(Y^2,\mu^{\otimes 2})}.$
 \end{corollary}

Recall that
 $K\in\cB(H_Y)$
is a~trace class operator iff there exist operators
 $K_{1},K_{2}\in\cHS(H_Y)$
such that $K=K_{1}K_{2}.$ Moreover,
 $\Tr\,K
 =\inn{K_{1}^{\ast}}{K_{2}}_{\cHS(H_Y)}.$
This fact, Theorem~\ref{hscalk}, and Corollary~\ref{hscalkcor}
imply the following lemma, in which elements of the $L^2$ space
are distinguished from their representatives. The element of
the $L^2$ space represented by a~function $f$ is denoted by $[f].$

 \begin{lemma}{\label{tr}}
 Let
  $K\in\cT\left(H_Y\right),$ $K =K_{1}K_{2},$
 where
  $K_{1},K_{2}\in\cHS\left(H_Y\right).$
 Let
  $[\cK_{1}],[\cK_{2}]\in L^2(Y^2,\mu^{\otimes 2})$
 be integral kernels of $K_{1},K_{2}.$ Then for any choice of
 representatives $\cK_1\in [\cK_1],$ $\cK_2\in [\cK_2]$ the function
 $\cK\colon Y\times Y\to\K$
 defined for \mbox{$\mu^{\otimes 2}$-a.a.} $(x,y)\in Y\times Y$ by
  \begin{equation}{\label{tr1}}
  \cK(x,y) =\int_{Y}\cK_{1}(x,z)\cK_{2}(z,y)\,\dx\mu(z)
  \end{equation}
 is \mbox{$\mu^{\otimes 2}$-square} integrable and it is an
 integral kernel of $K.$ The function $\cL\colon Y\to\K$ defined
 for \mbox{$\mu$-a.a.} $x\in Y$ by $\cL(x) =\cK(x,x)$ is
 \mbox{$\mu$-inte}grable. Moreover,
  \begin{equation}{\label{tr2}}
  \Tr\,K =\int_{Y}\cL(x)\dx\mu(x) \equiv\int_{Y}\cK(x,x)\,\dx\mu(x).
  \end{equation}
 \end{lemma}

 \begin{definition}{\label{ilkern}}
 Under the assumptions of Lemma~\ref{tr}, the function $\cK$ given
 by formula~\eqref{tr1} (for any choice of representatives $\cK_1,$
 $\cK_2$ of $[\cK_1],[\cK_2]$)
 is called a~\emph{product integral kernel of $K$}.
 \end{definition}

Notice that for $\mu$ being the Lebesgue measure on
$[0,1]\times[0,1]$ formula~\eqref{tr2} is valid, for example, if
$\cK$ is any continuous function.

In the following lemma, which follows from Lemma~\ref{tr},
the function $\cK_0$ need not be a~product integral
kernel of $K_0$ but the integral formula for the trace of $K_0$
still holds for $\cK_0.$

 \begin{lemma}{\label{parttr}}
 Let $k,n\in\N,$ $k<n,$ and let $\cK$ be a~product integral kernel
 of
  $K\in\cT\left(H_Y^{\otimes n}\right)
  \equiv\cT\left(L^2(Y^n,\mu^{\otimes n})\right).$
 Then the function
  $\cK_{0}\colon Y^{k}\times Y^{k}\to\K$
 defined for \mbox{$\mu^{\otimes 2k}$-a.a.}
  $(x^{\prime},y^{\prime})\in Y^{k}\times Y^{k}$
 by
  \begin{equation}{\label{parttr1}}
  \cK_{0}(x^{\prime},y^{\prime})
  =\int_{Y^{n-k}}
  \cK(x^{\prime},x^{\prime\prime},y^{\prime},x^{\prime\prime})\,
  \dx\mu^{\otimes(n-k)}(x^{\prime\prime})
  \end{equation}
 is
 \mbox{$\mu^{\otimes 2k}$-square} integrable and
 the integral operator $K_0$ with the kernel $\cK_0$
 belongs to
 $\cT\left(H_Y^{\otimes k}\right).$
 For every
  $\chi,\varphi\in H_Y^{\otimes k}$
 and every orthonormal basis $\set{\psi_i}_{i\in\N}$ of
  $H_Y^{\otimes (n-k)}$
 one has
 \begin{equation*}
 \inn{\chi}{K_0\varphi}_{H_Y^{\otimes k}}
 =\sum_{i=1}^{\infty} \inn{\chi\otimes\psi_i}{
 K(\varphi\otimes\psi_i)}_{H_Y^{\otimes n}}.
 \end{equation*}
 The function
  $\cL_{0}\colon Y^{k}\to\K$
 defined for \mbox{$\mu^{\otimes k}$-a.a.} $x^{\prime}\in Y^{k}$ by
  $\cL_0(x^{\prime}) =\cK_0(x^{\prime},x^{\prime})$
 is \mbox{$\mu^{\otimes k}$-inte}grable. Moreover,
  \begin{equation*}
  \int_{Y^{k}}\cK_{0}(x^{\prime},x^{\prime})
  \,\dx\mu^{\otimes k}(x^{\prime})
  \equiv\int_{Y^{k}}\cL_{0}(x^{\prime})
  \,\dx\mu^{\otimes k}(x^{\prime}) =\Tr\,K_{0} =\Tr\,K.
  \end{equation*}
 \end{lemma}

 \begin{corollary}{\label{red}}
 Under the assumptions of Lemma~\ref{parttr}, if
  $C\in\cB\left(H_Y^{\otimes k}\right)$
 then
  $\Tr\,CK_{0} =\Tr\,(C\otimes\Id^{\otimes(n-k)})K.$
 \end{corollary}

\section*{Acknowledgements}

This article presents the results of a part of the autor's MS
thesis~\cite{RadzkiUMK99} written in Institute of Physics, Nicolaus
Copernicus University, Toru\'{n}, under the supervision of
Professor Jan Ma\'{c}kowiak. The author wishes to express his
gratitude to Professor Ma\'{c}kowiak for helpful suggestions
and remarks. Professor Ma\'{c}kowiak prepared also, on his own
initiative, the English translation of appropriate parts of the
author's thesis, which was useful for the author in editing of the
present paper.

\end{document}